\documentclass[12pt,onecolumn,draftcls]{IEEEtran}
\IEEEoverridecommandlockouts
\usepackage{amsmath,epsfig}
\usepackage{amssymb}
\usepackage{amsmath}
\usepackage{cases} %onecolumn,
\usepackage{amssymb,amsmath,cite}
\usepackage{epsfig}
\usepackage{color}
\usepackage{slashbox}
\usepackage{slashbox}
\usepackage[linesnumbered,boxed]{algorithm2e}
\usepackage{bm}

\newcommand{\be}{\mathbf{e}}

\newcommand{\bv}{\mathbf{v}}

\newcommand{\bu}{\mathbf{u}}

\newcommand{\bI}{\mathbf{I}}

\newcommand{\by}{\mathbf{y}}
\newcommand{\bz}{\mathbf{z}}

\newcommand{\bH}{\mathbf{H}}

\renewcommand{\frac}{\dfrac}

\newcommand{\K}{{\cal K}}
\newcommand{\N}{{\cal N}}

\newcommand{\SI}{\mbox{SINR}}

\newtheorem{dingyi}{Definition~}%[section]
\newtheorem{dingli}{Theorem~}%[section]
\newtheorem{yinli}{Lemma~}%[section]
\newtheorem{lizi}{Example~}%[section]

\newenvironment{proof}[1][Proof]{\begin{trivlist}
\item[\hskip \labelsep {\bfseries #1}]}{\end{trivlist}}

\ifCLASSINFOpdf
\else
\fi
\begin{document}
  \title{Dynamic Spectrum Management: A Complete Complexity Characterization % for Multi-User Multi-Subcarrier Communication Systems
\author{Ya-Feng Liu}
%\thanks{Part of this work has been presented at the IEEE
%International Conference on Communications (ICC), Kyoto, Japan, June
%5-9, 2011\cite{ICC}. This work was supported in part by the China National Funds for Distinguished Young Scientists, Grant 11125107, the National Natural Science Foundation, Grant 10831006, and the CAS Grant kjcx-yw-s7-03, and in part by the Army Research Office, Grant
%W911NF-09-1-0279, and the National Science Foundation,
%Grant CMMI-0726336.} \thanks{Copyright (c) 2012 IEEE. Personal use of this material is permitted. However, permission to use this material for any other purposes must be obtained from the IEEE by sending a request to pubs-permissions@ieee.org.}
 \thanks{Y.-F.~Liu is with the State Key Laboratory
of Scientific and Engineering Computing, Institute of Computational
Mathematics and Scientific/Engineering Computing, Academy of
Mathematics and Systems Science, Chinese Academy of Sciences,
Beijing, 100190, China (e-mail:
{{yafliu}@lsec.cc.ac.cn}). This work is partially supported by the National Natural
Science Foundation of China under Grants 11671419, 11301516, 11331012, 11631013, and 11571221.}}
\maketitle

\begin{abstract}
\boldmath Consider a multi-user multi-carrier communication system where multiple users share multiple discrete subcarriers. To achieve high spectrum efficiency, the users in the system must choose their transmit power dynamically in response to fast channel fluctuations. Assuming perfect channel state information, two formulations for the spectrum management (power control) problem are considered in this paper: the first is to minimize the total transmission power subject to all users' transmission data rate constraints, and the second is to maximize the min-rate utility subject to individual power constraints at each user. It is known in the literature that both formulations of the problem are polynomial time solvable when the number of subcarriers is one and strongly NP-hard when the number of subcarriers are greater than or equal to three. However, the complexity characterization of the problem when the number of subcarriers is two has been missing for a long time. This paper answers this long-standing open question: both formulations of the problem are strongly NP-hard when the number of subcarriers is two.

%The proof is based on a polynomial time transformation from the MAX-2UNANIMITY problem.
%This paper shows that  We therefore answer an open question

\end{abstract}
\IEEEpeerreviewmaketitle

%\textbf{Further Questions:}
%\begin{itemize}
%  \item Why we want to work in the rate region? Is there any advantage? Make this point clear by using a small concrete example or intensive simulations?
%  \item Coding is important!!! Write the code in this week!
%\end{itemize}

\section{Introduction}

In multi-carrier (multi-tone) communication systems, the transmission frequency spectrum is partitioned into a number of orthogonal subcarriers on which parallel data can be simultaneously transmitted without causing interferences with each other. However, different users interfere with each other on the same subcarrier. The common examples of multi-carrier systems include wireless orthogonal frequency division multiplex (OFDM) systems (such as the 802.11) and wireline discrete multi-tone (DMT) systems (such as the digital subscriber line (DSL) system). In both of OFDM and DMT systems, a pair of discrete Fourier transform (DFT) and inverse discrete Fourier transform (IDFT) is used to effectively decompose the frequency-selective wideband channel into a group of non-selective narrowband subcarriers, which makes them robust against large delay spreads by preserving orthogonality in the frequency domain \cite{goldsmith2005wireless,tse2005fundamentals}.

Spectrum management, also called spectrum balancing or power control, is a central issue in the design of interference-limited multi-user multi-carrier communication systems. This is because in such systems the achievable data rate of each user depends not only on its own power allocation but also on the power allocation of all other users. The spectrum management problem in multi-user multi-carrier communication systems is often formulated as an optimization problem such as the system utility maximization problem subject to power budget constraints or the total power minimization problem subject to Quality-of-Service (QoS) constraints.

% Compared to the Nash game formulation, the optimization formulation can lead to a much better power allocation strategy. However, the optimization formulation is numerically difficult to solve, since the objective functions and/or the constraints are highly nonlinear and nonconvex.
%[[[\textbf{cite some recent references!}]]
The spectrum management problem in the interference-limited multi-user multi-carrier communication system has been extensively studied; see \cite{luo08jstpdynamic,hayashi2009spectrum,liu2014complexity,yu02jsacdistributed,luo06eurasipiwf,yu06tcomosb,cendrillon05icciterative,lui05icclowcomplexity,yu06tcomdual,papandriopoulos09titscale,luo2009duality,wang2012scale,tsiaflakis2014real,Sahu15iwf,tsiaflakis2010improved,moraes2014general,liu14wcldownlink,yu07ISIT,moraes2014spherical,liu2015iterative} and references therein. The authors of \cite{luo08jstpdynamic} showed that the problem (under various optimization models) is (strongly) NP-hard when the number of subcarriers is greater than or equal to three. They also identified several subclasses of the problem which are polynomial time solvable when the number of subcarriers is one, such as the min-rate utility maximization problem and the total power minimization problem. However, the complexity characterization of the problem for the case where the number of subcariers is two was missed for a long time in the literature. {This might be due to the following two reasons. \begin{itemize}
   \item [--] The standard way to prove an optimization problem is NP-hard is to establish a polynomial time reduction from a known NP-complete problem to its corresponding feasibility problem or decision problem \cite{Complexitybook,combook2,combook3,combinatorial}. Since there are a large number of NP-complete problems involving ``three'', one has a lot of choices to pick an NP-complete problem involving ``three'' (as did in \cite{luo08jstpdynamic}), establish a polynomial time reduction from it to the spectrum management problem with three subcarriers, and show that the problem is NP-hard. In contrast, there are few NP-complete problems involving ``two'' and this makes it hard to show the NP-hardness of the spectrum management problem when the number of subcarriers is two. In fact, there is usually a complexity gap between problems involving ``two'' and ``three'', i.e., 3-SATISFIABILITY, 3-DIMENSIONAL MATCHING, and 3-COLORABILITY problems are all NP-complete but 2-SATISFIABILITY, 2-DIMENSIONAL MATCHING, and 2-COLORABILITY problems are all polynomial time solvable.
   \item [--] Exploring the complexity of an optimization problem typically involves switching back and forth between trying to develop a polynomial time algorithm for the problem and trying to prove it NP-hard, until one of them succeeds. Since many problems involving ``two'' are polynomial time solvable (as mentioned above), this might give a wrong sense that the spectrum management problem is polynomial time solvable when the number of subcarriers is two, which makes it hard to characterize the complexity of the problem.
 \end{itemize}}

%lies in the big complexity gap between problems involving ``two'' and ``three''. There are a lot  but there are  For instance,

The complexity results in \cite{luo08jstpdynamic} suggest that there is no polynomial time algorithms which can solve the general spectrum management problem to global optimality (unless P$=$NP), and determining an approximately optimal or locally optimal spectrum management strategy in polynomial time is more realistic in practice (especially when a very fast responsiveness is required \cite{tsiaflakis2014real}). Therefore, various (heuristic) algorithms \cite{luo08jstpdynamic,hayashi2009spectrum,liu2014complexity,yu02jsacdistributed,luo06eurasipiwf,yu06tcomosb,cendrillon05icciterative,lui05icclowcomplexity,yu06tcomdual,papandriopoulos09titscale,luo2009duality,wang2012scale,tsiaflakis2014real,Sahu15iwf,tsiaflakis2010improved,moraes2014general,liu14wcldownlink,yu07ISIT,moraes2014spherical}, including iterative water-filling algorithms, dual decomposition algorithms, and successive convex/concave approximation algorithms, have been proposed for solving the problem. %For instance, \cite{yu02jsacdistributed} proposed an iterative water-filling algorithm for solving the power allocation problem. Also, dual decomposition based algorithms and successive convex approximation algorithms were proposed in \cite{yu06tcomosb,cendrillon05icciterative,lui05icclowcomplexity,yu06tcomdual} and \cite{papandriopoulos09titscale}.

In this paper, we focus on the characterization of the computational complexity status of the spectrum management problem for the multi-user multi-carrier communication system. In particular, we consider two formulations of the problem. The first one is the problem of minimizing the total transmission power subject to all users' QoS constraints. The second one is the problem of maximizing the minimum rate among all users while respecting the total transmission power constraint of each user. {The main contribution of this paper is to answer a long-standing open question: both aforementioned formulations of the spectrum management problem are strongly NP-hard when the number of subcarriers is two. {The developed techniques in this paper can be extended to show the (strong) NP-hardness of other related optimization problems involving ``two'' arising from signal processing and wireless communications.}
%, like the linear transceiver design problem in the multi-user single carrier multi-input multi-output (MIMO) interference channel where each transmitter and receiver is equipped with two antennas.}. %The proof is based on a polynomial time transformation from the MAX-2UNANIMITY problem.
}
%
%This paper considers the total power minimization problem subject to QoS constraints of all users in multi-user multi-carrier communications systems.
%The goal of this paper is to design a near-optimal yet low-complexity algorithm with guaranteed convergence (to a stationary point or a KKT solution) for the above optimization problem.

%The Nash game formulation often lead to a poor power allocation strategy while the optimization formulation can achieve much better performance than the Nash game formulation.
%
%In particular for the optimization formulation, the design objective function and/or the constraints are often nonconvex, which makes the optimization problem numerically difficult to solve.

%\textbf{Notations.} We adopt the following notations in this paper. Lowercase boldface and uppercase
%boldface are used for vectors and matrices, respectively.~We use $\bx_1\circ\bx_2$ to represent the Hadamard product of two vectors $\bx_1$ and $\bx_2.$
%The spectral radius of a matrix $\bA$ is denoted by $\rho(\bA).$
%We use $\be$ to represent the vector with all entries being one, $\be_k$ to represent the vector with all entries being zero except its $k$-th entry being one, and $\bI$ to represent the identity matrix of an
%appropriate size, respectively. %We also use $\be_k$ to denote the

\section{Problem Formulation}

Consider a multi-user multi-carrier communication system, where there are $K$ users (transmitter-receiver pairs) sharing $N$ discrete subcarriers. Denote the set of users and the set of subcarriers by $\K=\left\{1,2,\ldots,K\right\}$ and $\N=\left\{1,2,\ldots,N\right\}$, respectively. For any $k\in\K$ and $n\in\N$, suppose $s_k^n\in\mathbb{C}$ to be the symbol that transmitter $k$ wishes to send to receiver $k$ on subcarrier $n$, then the received signal $\hat s_k^n$ at receiver $k$ on subcarrier $n$ can be expressed by
$$\hat s_k^n=\sum_{j\in\K}h_{k,j}^n s_j^n+ z_k^n,$$ where $h_{k,j}^n\in\mathbb{C}$ is the channel coefficient between the $j$-th transmitter and the $k$-th receiver on subcarrier $n$ and $z_k^n\in\mathbb{C}$ is the additive white Gaussian
noise (AWGN) with distribution $\cal{CN}$$({0}, \eta_k^n).$
%
%Let
%\begin{itemize}
%  \item [-] $p_k^n$ $(k=1,2,...,K,~n\in\N)$ denote the transmission power of user $k$ on subcarrier $n;$
%  \item [-] $\alpha_{kj}^n$ $(k,~j=1,2,...,K,~n\in\N)$ denote the channel gain from transmitter $j$ to receiver $k$ on subcarrier $n;$
%  \item [-] $\eta_k^n$ $(k=1,2,...,K,~n\in\N)$ denote the noise power at receiver $k$ on subcarrier $n.$
%\end{itemize}
Denoting the power of $s_k^n$ by $p_k^n$; i.e., $p_k^n:=|s_k^n|^2$, the received power at receiver $k$ on subcarrier $n$ is given by
$$\sum_{j\in\K}g_{k,j}^np_j^n+\eta_k^n,~k\in\K,~n\in\N,$$ where $g_{k,j}^n:=|h_{k,j}^n|^2$ %$(k,~j=1,2,...,K,~n\in\N)$
stands for the channel gain between the $j$-th transmitter and the $k$-th receiver on subcarrier $n.$ Treating interference as noise, we can write the SINR of receiver $k$ on subcarrier $n$ as
$$\SI_k^n=\frac{g_{k,k}^np_k^{n}}{\displaystyle\sum_{j\neq k}g_{k,j}^np_j^n+\eta_k^n},~k\in\K,~n\in\N,$$ and transmitter $k$'s achievable data rate $R_k$ {(nats/sec)} as \begin{equation}\label{rk}R_{k}=\sum_{n\in\N}{\ln}\left(1+\frac{g_{k,k}^np_k^{n}}{\displaystyle\sum_{j\neq k}g_{k,j}^np_j^n+\eta_k^n}\right),~k\in\K.\end{equation}

In this paper, we consider the following two formulations of the spectrum management problem: % for the multi-user multi-carrier communication system: %is often formulated as either
\begin{equation}\label{problem}
 \begin{array}{cl}
\displaystyle \min_{\left\{p_k^n\right\}} & \displaystyle \sum_{k\in\K}\sum_{n\in\N}p_k^{n}\\[5pt]%+\lambda\sum_{g=1}^{G}\|\left(\sum_{k=1}^K\bp_k\right)_{g}\|_{2} \\[20pt]
\text{s.t.} & R_{k}\geq \gamma_{k},~k\in\K, \\%[3pt]
          & \displaystyle \sum_{n\in\N} p_k^n \leq \bar p_k,~k\in\K,\\
        & p_k^n\geq 0,~k\in\K,~n\in\N,
    \end{array}
\end{equation}
 and \begin{equation}\label{problemmin}
 \begin{array}{cl}
\displaystyle \max_{\left\{p_k^n\right\}} & \displaystyle \min_{k\in\K}\left\{R_{k}\right\}\\[5pt]%+\lambda\sum_{g=1}^{G}\|\left(\sum_{k=1}^K\bp_k\right)_{g}\|_{2} \\[20pt]
\text{s.t.} & \displaystyle \sum_{n\in\N}p_k^{n}\leq \bar p_{k},~k\in\K, \\%[3pt]
        & p_k^n\geq 0,~k\in\K,~n\in\N,
    \end{array}
\end{equation}where $\gamma_k$ is the desired transmission rate target of user $k$ and $\bar p_k$ is the power budget of transmitter $k.$ Formulation \eqref{problem} minimizes the total transmission power of all users on all subcarriers and  formulation \eqref{problemmin} maximizes the minimum transmission rate among all users.% subject to power budget constraints.%, where $w_k>0$ ($k\in\K$) are positive weights and $\bar p_k$ ($k\in\K$) is the total power budget for user $k.$

\section{Complexity Analysis}
In this section, we first briefly introduce complexity theory in Section \ref{introduction}. Then, we review existing complexity results of problems \eqref{problem} and \eqref{problemmin} and show that both problems are strongly NP-hard when the number of subcarriers is two in Section \ref{subsec:newcom}. {Finally, we extend the developed techniques to show the (strong) NP-hardness of other related optimization problems involving ``two'' arising from signal processing and wireless communications in Section \ref{sec-extension}.}

\subsection{A Brief Introduction to Complexity Theory}\label{introduction}

In computational complexity theory{\cite{Complexitybook,combook2,combook3,combinatorial}}, a problem is said to be NP-hard if it is at least as hard as any problem in the class NP (problems that are solvable in Nondeterministic Polynomial time). NP-complete problems are the hardest problems in NP in the sense that if any NP-complete problem is solvable in polynomial time, then each problem in NP is solvable in polynomial time. A problem is strongly NP-hard (strongly NP-complete) if it is NP-hard (NP-complete) and it cannot be solved by a pseudo-polynomial time algorithm. {An algorithm that solves a problem is called a \emph{pseudo-polynomial} time algorithm if its time complexity function is bounded above by a polynomial function related to both of the {length} and the numerical values of the given data of the problem. This is in contrast to the polynomial time algorithm whose time complexity function depends only on the length of the given data of the problem.} It is widely believed that there can not exist a polynomial time algorithm to solve any NP-complete, NP-hard, or strongly NP-hard problem (unless P$=$NP). %Thus, once an optimization problem is shown to be NP-hard, we can no longer insist on having an efficient algorithm that can find its global optimum in polynomial time.

The standard way to prove an optimization problem is NP-hard is to establish the NP-hardness of its corresponding feasibility problem or decision problem. The latter is the problem to decide if the global minimum (maximum) of the optimization problem is below (above) a given threshold or not. To show a decision problem $\mathcal{P}_2$ is NP-hard,
we usually follow three steps: 1) choose a suitable NP-complete decision problem $\mathcal{P}_1;$ 2) construct a {polynomial
time} transformation from any instance of $\mathcal{P}_1$ to an instance of $\mathcal{P}_2;$
3) prove under this transformation that any instance of problem
  $\mathcal{P}_1$ is true if and only if the constructed instance of problem $\mathcal{P}_2$ is true. See \cite{Complexitybook,combinatorial,combook2,combook3} for more on complexity theory.

\subsection{Strong NP-Hardness of Problems \eqref{problem} and \eqref{problemmin} when $N=2$}\label{subsec:newcom}
%\subsection{A Brief Summary}\label{review}
Both problems \eqref{problem} and \eqref{problemmin} are polynomial time solvable when $N=1.$ More specifically, when $N=1,$ problem \eqref{problem} is equivalent to
\begin{equation}\label{problemN=1}
 \begin{array}{cl}
\displaystyle \min_{\left\{p_k\right\}} & \displaystyle \sum_{k\in\K}p_k\\[3pt]
\text{s.t.} & {g_{k,k}p_k}\geq \left(\exp{\left(\gamma_{k}\right)}-1\right)\left(\displaystyle\sum_{j\neq k}g_{k,j}p_j+\eta_k\right),~k\in\K, \\[3pt]
        & \bar p_k\geq p_k\geq 0,~k\in\K,
    \end{array}
\end{equation} which is a linear program (solvable in polynomial time). When $N=1,$ problem \eqref{problemmin} reduces to
\begin{equation}\label{problemminN=1}
 \begin{array}{cl}
\displaystyle \max_{\tau,\,\left\{p_k\right\}} & \displaystyle \tau \\[5pt]
\text{s.t.} & g_{k,k}p_k\geq \tau\left(\displaystyle\sum_{j\neq k}g_{k,j}p_j+\eta_k\right),~k\in\K,\\[5pt]
& \bar p_k\geq p_k\geq 0,~k\in\K,
    \end{array}
\end{equation}which is polynomial time solvable by using a binary search on $\tau;$ see \cite[Theorem 2]{luo08jstpdynamic}. In fact, both problems \eqref{problem} and \eqref{problemmin} are also polynomial time solvable (by the water-filling algorithm) when $K=1$ (i.e., there is only a single user in the system) \cite[Theorem 4.1]{liu2014complexity}.

Problems \eqref{problem} and \eqref{problemmin} become computationally intractable when the number of subcarriers is greater than or equal to three. In particular, it is shown in \cite[Theorem 2]{luo08jstpdynamic} that problem \eqref{problem} is strongly NP-hard when $N\geq 3.$ By using the same argument as in the proof of \cite[Theorem 2]{luo08jstpdynamic}, one can also show the strong NP-hardness of problem \eqref{problemmin} with $N\geq 3.$ However, the complexity characterization of problems \eqref{problem} and \eqref{problemmin} with $N=2$ has been missing for a long time in the literature. In this subsection, we answer this open question and show that both of problems \eqref{problem} and \eqref{problemmin} remain strongly NP-hard when $N=2.$ %We therefore answer a long-standing open question whether problems \eqref{problem} and \eqref{problemmin} with $N=2$ are

The NP-hardness proof of problems \eqref{problem} and \eqref{problemmin} for the case $N=2$ is based on a polynomial time
{reduction} from the MAX-2UNANIMITY~problem, which was first introduced in \cite{liu11tspbeamforming}. To describe the problem, we first
define the UNANIMITY property of a disjunctive clause. {Recall
that for a given set of Boolean variables, a literal is defined as
either a Boolean variable or its negation, while a disjunctive
clause refers to a logical expression consisting of the logical
``OR" of literals.}

\begin{dingyi}[UNANIMOUS] For a given truth assignment to a set of Boolean variables,
a disjunctive clause is said to be {unanimous} if all literals
in the clause have the same value (whether it is the {True} or
the {False} value). %Otherwise it is said to be {NAE} (Not-All-Equal).
\end{dingyi}

%Notice that a disjunctive clause must be satisfied if it is to have
%the NAE property.
%We now define some \emph{decision problems} over
%Boolean variables.
%MAX-UNANIMITY problem, MAX-2UNANIMITY problem and
%NAE-SAT problem are defined as follows:
\begin{dingyi}[MAX-2UNANIMITY]
 %is as follows:
Given a positive integer $M$ and $m$ disjunctive clauses defined over $n$
Boolean variables, where the number of literals in each clause is $2,$ the {MAX-2UNANIMITY}~problem is to check whether there exists a truth assignment
such that the number of {unanimous} disjunctive clauses is at
least $M$.  %the corresponding problem is called {MAX-2UNANIMITY.}
\end{dingyi}

{\begin{lizi}\label{example}Given Boolean variables $x_1,x_2,x_3,x_4,$ define $c_{1}=x_1\vee {\bar x}_2,~c_2=x_1\vee {x}_3,$ $c_3=\bar x_2\vee\bar x_4,~c_4=\bar x_3\vee x_4.$ All of $x_1,x_2,x_3,x_4$ and their negations $\bar x_1,\bar x_2,\bar x_3,\bar x_4$ are literals; all of $c_1,c_2,c_3,c_4$ are disjunctive clauses; if we set $x_1=1, x_2=0, x_3=1, x_4=0,$ then all of $c_1, c_2, c_3, c_4$ are unanimous (or satisfied unanimously); the clauses $c_1,c_2,c_3,c_4$ defined on $x_1,x_2,x_3,x_4$ along with some given positive $M$ is an instance of the MAX-2UNANIMITY problem. %It can be checked that there does not exist a truth assignment such that all $4$ clauses are satisfied unanimously; while there exists a truth assignment such that $3$ clauses are satisfied unanimously.
\end{lizi}}

\begin{yinli}[\!\!\cite{liu11tspbeamforming}]\label{MAX-Full}
MAX-2UNANIMITY problem is NP-complete.
\end{yinli}

%To show the NP-hardness of problems \eqref{problem} and \eqref{problemmin}, we also need Lemma \ref{lemma_small_power} Appendix \ref{app:lemma}.%, whose proof is relegated to .%, which characterizes the feasible solutions of %a special case of problem \eqref{problem} as follows:

%In the above problem, there are only $2$ users and $2$ subcarriers in the system.
%also need the following lemma.

%Since both of the above points are feasible to problem \eqref{problem}, they
%
%is less than or equal to the one of problem \eqref{problem}.

We are now ready to prove our main results.

\begin{dingli}\label{thm_complexity}
  Problem \eqref{problem} is strongly NP-hard when {$N = 2$}.
\end{dingli}

\begin{proof} %Problem \eqref{problem} when $N=1$ (see problem \eqref{problemN=1}) is essentially a linear program, which can be solved efficiently in polynomial time. Next, we show problem \eqref{problem} is strongly NP-hard when $N\geq 2.$
Given any instance of the MAX-2UNANIMITY problem with clauses
$c_1,c_2,\ldots,c_m$ defined over Boolean variables $x_1,x_2,\ldots,x_n$
and an integer $M,$ we construct below a multi-user multi-carrier interference channel with $2n+m$ users and $2$ subcarriers, where the Boolean variable $x_i~(i=1,2,\ldots,n)$ corresponds to a pair of users, including user $i$ (called ``variable user'') and user $n+i$ (called ``auxiliary variable user''); each clause $c_j~(j=1,2,\ldots,m)$ corresponds to user $2n+j$ (called ``clause user''). %Therefore, there are $K=2n+m$ users and $N=2$ subcarriers in the constructed system.
Hence, ${\K}=\{1,2,\ldots,2n+m\}$ and $\N=\left\{1,2\right\}.$

{Next, we construct channel parameters for all $2n+m$ users on $2$ subcarriers. Before going into very details, let us first give a high level preview of the construction. More specifically, we first construct channel parameters of all users associated with Boolean variables (i.e., ``variable users'' and ``auxiliary variable users'') such that the only ways for each pair of the users to satisfy their transmission rate requirements are that one user transmits full power on one subcarrier and the other user transmits full power on the other subcarrier. Then, we construct channel parameters of all users associated with clauses (i.e., ``clause users'') such that each ``clause user'' suffers interferences from only two ``variable users'' and/or ``auxiliary variable users'', whose types (i.e., ``variable user'' or ``auxiliary variable user'') and indices are determined by the two literals appearing in the corresponding clause. In addition, we construct channel parameters of ``clause users'' such that the required total transmission power for the user to satisfy its transmission rate constraint when the corresponding clause is unanimous is strictly less  than the one when the corresponding clause is not. In this way, the total transmission power for all users to satisfy their transmission rate constraints depends on the number of unanimous clauses and less total transmission power is needed if and only if more clauses are satisfied unanimously.}

%the channel parameters such that all users associated with clauses do not cause
%
%first construct the channel parameters of $2n$ users corresponding to the $n$ Boolean variables such that the power allocation of these $2n$ users to satisfy the required transmission rate
%
%The basic idea of constructing an instance of problem \eqref{problem} based on the given instance of the MAX-2UNANIMITY problem

{We now construct the direct-link and crosstalk channel gains among these $2n+m$ users on $2$ subcarriers. %based on the given instance of the MAX-2UNANIMITY problem.
The direct-link channel gains of all users on two subcarriers are set to be $$g_{k,k}^1=g_{k,k}^2=1,~k\in\K.$$ The corresponding crosstalk channel gains on $2$ subcarriers are: for user $k=1,2,\ldots,n,$ set $$g_{k,n+k}^1=g_{k,n+k}^2=1~\text{and}~g_{k,\ell}^1=g_{k,\ell}^2=0, \forall~\ell\in\K\setminus\left\{k, n+k\right\};$$ for user $k=n+1,n+2,\ldots,2n,$ set $$g_{k,k-n}^1=g_{k,k-n}^2=1~\text{and}~g_{k,\ell}^1=g_{k,\ell}^2=0,~\forall~\ell\in\K\setminus\left\{k-n, k\right\};$$ for user $k=2n+1,2n+2,\ldots,2n+m,$ set $g_{k,\ell}^1=g_{k,\ell}^2=0$ for all $\ell\in\K$ except %$g_{k,\ell}^1=g_{k,\ell}^2=1$ if $x_{\ell}$ appears in $c_{k-2n}$
%and $g_{k,n+\ell}^1=g_{k,n+\ell}^2=1$ if $\bar x_{\ell}$ appears in $c_{k-2n}.$~
\begin{equation*}\label{matrix}\begin{array}{l} \left\{
\begin{array}{ll}g_{k,\ell}^1=g_{k,\ell}^2=1,&\mbox{if $x_{\ell}$ appears in $c_{k-2n}$;}\\
g_{k,n+\ell}^1=g_{k,n+\ell}^2=1,&\mbox{if $\bar x_{\ell}$ appears in $c_{k-2n}.$}\end{array}\right.
\end{array}\end{equation*}
%\begin{equation}\label{matrix}\begin{array}{l} \left\{
%\begin{array}{ll}g_{k,\ell}^1=g_{k,\ell}^2=1,&\mbox{if $x_{\ell}$ appears in $c_{k-2n}$$\alpha_{\pi(k)}={x}_{\ell}$ for some
%$l;$}\\
%g_{k,n+\ell}^1=g_{k,n+\ell}^2=1,&\mbox{if $\alpha_{\pi(k)}={\bar x}_{\ell}$ for some
%$\ell,$}\end{array}\right.\\[15pt]
%\left\{\begin{array}{ll}g_{k,\ell}^1=g_{k,\ell}^2=1,&\mbox{if $\beta_{\rho(k)}=x_{\ell}$ for some $\ell;$}\\
%g_{k,n+\ell}^1=g_{k,n+\ell}^2=1,&\mbox{if $\beta_{\rho(k)}={\bar x}_{\ell}$ for some
%$\ell,$}\end{array}\right.%\\[15pt]%
%%{\bH}_{kj}=\left\{\begin{array}{ll}\bH_{F},&\mbox{if $\gamma_{\tau(k)}=x_j$ for some $j;$}\\
%%\bH_{E},&\mbox{if $\gamma_{\tau(k)}={\bar x}_j$ for some
%%$j,$}\end{array}\right.
%\end{array}\end{equation} where $c_{k-2n}=\alpha_{\pi(k)}\vee \beta_{\rho(k)}$,
%$\alpha$, and $\beta$ are taken from $\{x,{\bar x}\},$ and
%$\pi$ and $\rho$ are mappings from $\{2n+1,2n+2,...,2n+m\}$ to
%$\{1,2,...,n\}$.
%Set the noise power of all users on two subcarriers be $1.$
Set $\eta_k^n=1$ for all $k\in\K$ and $n\in\N.$} Then, the transmission rate expressions of all users are: for $i=1,2,\ldots,n,$ %let
\begin{equation}\label{variablerate1}
 R_{i}=\ln\left(1+\frac{p_{i}^1}{1+p_{n+i}^1}\right)+\ln\left(1+\frac{p_{i}^2}{1+p_{n+i}^2}\right)
 \end{equation} and
 \begin{equation}\label{variablerate2}
 R_{n+i}=\ln\left(1+\frac{p_{n+i}^1}{1+p_{i}^1}\right)+\ln\left(1+\frac{p_{n+i}^2}{1+p_{i}^2}\right);
\end{equation}
for $j=1,2,\ldots,m,$ %let
{\begin{equation}\label{clauserate}R_{2n+j}=\left\{\!\!\!\!\!\!\!
  \begin{array}{ll}
&\displaystyle \ln\left(1+\frac{p_{2n+j}^1}{1+p_{i_1}^1+p_{i_2}^1}\right)+\ln\left(1+\frac{p_{2n+j}^2}{1+p_{i_1}^2+p_{i_2}^2}\right),\,\text{if}~c_j=x_{i_1}\vee x_{i_2};\\[15pt]
&\displaystyle \ln\left(1+\frac{p_{2n+j}^1}{1+p_{i_1}^1+p_{n+i_2}^1}\right)+\ln\left(1+\frac{p_{2n+j}^2}{1+p_{i_1}^2+p_{n+i_2}^2}\right),\,\text{if}~c_j=x_{i_1}\vee \bar x_{i_2};\\[15pt]
&\displaystyle \ln\left(1+\frac{p_{2n+j}^1}{1+p_{n+i_1}^1+p_{i_2}^1}\right)+\ln\left(1+\frac{p_{2n+j}^2}{1+p_{n+i_1}^2+p_{i_2}^2}\right),\,\text{if}~c_j=\bar x_{i_1}\vee x_{i_2};\\[15pt]
&\displaystyle \ln\left(1+\frac{p_{2n+j}^1}{1+p_{n+i_1}^1+p_{n+i_2}^1}\right)+\ln\left(1+\frac{p_{2n+j}^2}{1+p_{n+i_1}^2+p_{n+i_2}^2}\right),\,\text{if}~c_j=\bar x_{i_1}\vee \bar x_{i_2}.
  \end{array}\right.
\end{equation}}In the above, each user $k$ is associated with two variables $p_k^1$ and $p_k^2$ for $k\in\K;$ each ``variable user'' $i$ suffers interference from ``auxiliary variable user'' $n+i$ on both subcarriers $1$ and $2;$ each ``auxiliary variable user'' $n+i$ suffers interference from ``variable user'' $i$ on both subcarriers $1$ and $2;$ each ``clause user'' $2n+j$ suffers interference from ``variable user'' $i_1$ and $i_2$ and/or ``auxiliary variable user'' $n+i_1$ and $n+i_2,$ where $c_j$ contains literals of $x_{i_1}$ and $x_{i_2}.$ To make the construction of transmission rate expressions clear, an illustrative
example is given in Appendix \ref{app-example}.

%=\alpha_{i_1}\vee \beta_{i_2}$ and $\alpha,\beta\in\left\{x,\bar x\right\}.$

%[[[\textbf{Talk about communication background!}]]]

Moreover, let $\gamma_{k}=\ln2$ and $\bar p_k=1$ for all $k\in\K.$ Then, the constructed instance of problem \eqref{problem} is
\begin{equation}\label{constructed}
\begin{array}{cl}
\displaystyle\min_{\left\{p_k^n\right\}} &\displaystyle \sum_{k\in\K}\sum_{n\in\N}p_k^n\\[10pt]%\leq 2n+M+4(m-M)(\sqrt{2}-1)
\mbox{s.t.}&\displaystyle R_{k}\geq \ln2,\,k\in\cal K,\\[1pt]
&\displaystyle p_k^1+p_k^2 \leq 1,\,k\in\K,\\[3pt]
&\displaystyle p_k^n\geq 0,\,k\in\K,~n\in\N,
\end{array}
\end{equation}where $R_k$ are given in \eqref{variablerate1}, \eqref{variablerate2}, and \eqref{clauserate}. {The variable correspondence between the MAX-2UNANIMITY problem and problem \eqref{constructed} is listed as Table \ref{relation}.}

{\begin{table}[!h]
\tabcolsep 5mm \caption{Variable Correspondence between MAX-2UNANIMITY Problem and Problem \eqref{constructed}} \label{relation}%\vspace{-0.4cm}
%\begin{center}
{\begin{tabular}{ll}
\hline \multicolumn{1}{l}{MAX-2UNANIMITY problem}&\multicolumn{1}{l}{Problem \eqref{constructed}}
\\ \hline
$m$ clauses defined over $n$ variables & $2n+m$ users communicate over $2$ subcarriers \\
Boolean variable $x_i$   &power allocation variables $p_i^1,p_i^2~p_{n+i}^1,p_{n+i}^2$ of users $i$ and $n+i$\\
clause $c_j$      &transmission rate $R_{{{2n+j}}}$ in \eqref{clauserate} of user $2n+j$ \\
literal $x_i$ appears in clause $c_j$&  user $i$ causes interference $p_i^1$ and $p_i^2$ to user $2n+j$ on $2$ subcarriers\\
literal $\bar{x}_i$ appears in clause $c_j$&     user $n+i$ causes interference $p_{n+i}^1$ and $p_{n+i}^2$ to user $2n+j$ on $2$ subcarriers\\
\hline
\end{tabular}}
%\end{center}
\end{table}}

{We claim that the transformation from the MAX-2UNANIMITY problem to problem \eqref{constructed} can be performed in polynomial time. The number of users and the number of subcarriers in problem \eqref{constructed} are $2n+m$ and $2,$ respectively. Hence, the size of problem \eqref{constructed} is bounded above by a polynomial (linear) function of the size of the MAX-2UNANIMITY instance. Moreover, the construction of channel parameters/transmission rate expressions for all $2n+m$ users is straightforward, i.e., for $i=1,2,\ldots,n,$ transmission rate expressions of users $i$ and $n+i$ (associated with variable $i$) are explicitly given in \eqref{variablerate1} and \eqref{variablerate2}; and for $j=1,2,\ldots,m,$ transmission rate expression of user $2n+j$ (associated with clause $j$) is explicitly given in \eqref{clauserate}. Therefore, the transformation from the MAX-2UNANIMITY problem to problem \eqref{constructed} can be performed in polynomial time.}

%We first consider the general problem \eqref{problem} where there are $K$ users and $N$ subcarriers. In this case, the total number of unknown variables of problem \eqref{problem} is $KN;$ and the total number of parameters in problem \eqref{problem} is $K(KN+N+2),$ including $K^2N$ of parameters $\left\{g_{k,j}^n\right\},$ $KN$ of parameters $\left\{\sigma_k^n\right\},$ $K$ of parameters $\left\{\gamma_k\right\},$ and $K$ of parameters $\left\{\bar p_k\right\}.$ Specializing this to problem \eqref{constructed}, there are $2(2n+m)$ unknown variables and $2(2n+m)(2n+m+2)$ parameters.
%
%Since there are $2n+m$ users and $2$ subcarriers in problem \eqref{constructed}, its total number of design variables are $2(2n+m).$ For each

Next, we show that there exists a truth assignment such that at least $M$
clauses are satisfied unanimously for the given MAX-2UNANIMITY
instance if and only if the optimal value of problem \eqref{constructed} is less than or equal to $2n+M+4(\sqrt{2}-1)(m-M).$

If there exists a truth assignment such that $M$ clauses in the
MAX-2UNANIMITY problem are unanimous, we claim that the optimal value of problem \eqref{constructed} is less than or equal to $2n+M+4(\sqrt{2}-1)(m-M).$ Let $\left\{x_i\right\}$ be the truth assignment such that $M$ clauses are unanimous in the MAX-2UNANIMITY problem. We set
$$p_{i}^1=p_{n+i}^2=1-x_i,~p_{i}^2=p_{n+i}^1=x_i,~i=1,2,\ldots,n.$$ With this, we can simply check that
$R_{k}\geq \ln 2$ for all~$k=1,2,\ldots,2n.$ Furthermore, we consider transmission rate requirements of the ``clause variable'' $2n+j$ with $j=1,2,\ldots,m.$

\begin{itemize}
  \item [-] If the clause $c_j$ is unanimous, then we have either
  $$R_{2n+j}=\ln\left(1+{p_{2n+j}^1}\right)+\ln\left(1+\frac{p_{2n+j}^2}{3}\right)$$
  or $$R_{2n+j}=\ln\left(1+\frac{p_{2n+j}^1}{3}\right)+\ln\left(1+{p_{2n+j}^2}\right).$$ In either cases, we can use a total transmission power of $1$ to make
  $R_{2n+j}\geq \ln2$ satisfied (by setting $\left(p_{2n+j}^1,p_{2n+j}^2\right)^T=\left(1,0\right)^T$ in the former case and $\left(p_{2n+j}^1,p_{2n+j}^2\right)^T=\left(0,1\right)^T$ in the latter case).
  \item [-] If the clause $c_j$ is not unanimous, then we must have
  $$R_{2n+j}=\ln\left(1+\frac{p_{2n+j}^1}{2}\right)+\ln\left(1+\frac{p_{2n+j}^2}{2}\right).$$ In this case, we can use a total transmission power of $4\left(\sqrt{2}-1\right)$ to make $R_{2n+j}\geq \ln2$ satisfied (by setting $p_{2n+j}^1=p_{2n+j}^2=2\left(\sqrt{2}-1\right)$).
\end{itemize}
As a result, if there exists a truth assignment such that at least $M$ clauses are satisfied unanimously, then the optimal value of problem \eqref{constructed} is less than or equal to $2n+M+4(\sqrt{2}-1)(m-M).$

For the converse part, assuming that the optimal value of problem \eqref{constructed} is less than or equal to $2n+M+4(\sqrt{2}-1)(m-M),$
we claim that at least $M$ clauses can be made unanimous. %according to Lemma \ref{lemma_small_power}, the optimal power allocation strategy to satisfy $R_{2i-1}\geq \ln2$ and $R_{2i}\geq \ln2$ is %at least $2$ and the power allocation vector must satisfying
%$$\left(p_{2i-1}^1,p_{2i-1}^2,p_{2i}^1,p_{2i}^2\right)^T=\left(1,0,0,1\right)$$ or
%$$\left(p_{2i-1}^1,p_{2i-1}^2,p_{2i}^1,p_{2i}^2\right)^T=\left(0,1,1,0\right).$$
%Since ``variable users'' and ``auxiliary variable users'' cause interferences to ``clause users'',
It follows from Lemma \ref{lemma_small_power} {in Appendix \ref{app:lemma}} that, for $i=1,2,\ldots,n,$ the optimal solution of problem \eqref{constructed} must be $$\left(p_{i}^1,p_{i}^2,p_{n+i}^1,p_{n+i}^2\right)^T=\left(1,0,0,1\right)$$ or
$$\left(p_{i}^1,p_{i}^2,p_{n+i}^1,p_{n+i}^2\right)^T=\left(0,1,1,0\right).$$
 This, together with \eqref{clauserate}, implies that the received total interferences at user $2n+j$ must be exactly $2$ for all $j=1,2,\ldots,m.$ More specifically, there might be two cases:
\begin{itemize}
  \item [-] Case 1: the received interference at user $2n+j$ is equal $2$ on one subcarrier and is equal to $0$ on the other one;
  \item [-] Case 2: the received interference at user $2n+j$ is equal to $1$ on both subcarriers.
\end{itemize}
If Case 1 happens for user $2n+j$, then the required total transmission power satisfying $R_{2n+j}\geq \ln2$ is at least $1;$ while if Case 2 happens for user $2n+j$, then the required total transmission power satisfying $R_{2n+j}\geq \ln2$ is at least $4(\sqrt{2}-1).$ By the assumption that the optimal value of problem \eqref{constructed} is less than or equal to $2n+M+4(\sqrt{2}-1)(m-M),$ we know that Case 1 must happen at least $M$ times (Case 2 cannot happen more than $m-M$ times). Moreover, it can be checked that $$x_i=1-p_{i}^1,~i=1,2,\ldots,n$$ is a truth assignment
which makes at least $M$ clauses in the MAX-2UNANIMITY problem satisfied unanimuously.

%Finally, the transformation from the MAX-2UNANIMITY problem to problem \eqref{constructed} can be finished in polynomial time.
Since the MAX-2UNANIMITY
problem is NP-complete (cf. Lemma \ref{MAX-Full}), we conclude that the problem of checking the optimal value of problem \eqref{constructed} is less than or equal to $2n+M+4(\sqrt{2}-1)(m-M)$ is strongly NP-hard. Hence, problem \eqref{problem} is strongly NP-hard.
\end{proof}

\begin{dingli}\label{thm-min}
  Problem \eqref{problemmin} is strongly NP-hard when {$N = 2$}. %\begin{equation}\label{problemmin}
% \begin{array}{cl}
%\displaystyle \max_{\left\{p_k^n\right\}} & \displaystyle \min_{k\in\K}\left\{R_{k}\right\}\\[5pt]%+\lambda\sum_{g=1}^{G}\|\left(\sum_{k=1}^K\bp_k\right)_{g}\|_{2} \\[20pt]
%\text{s.t.} & \displaystyle \sum_{n\in\N}p_k^{n}\leq \bar p_{k},~k\in\K, \\%[3pt]
%        & p_k^n\geq 0,~k\in\K,~n\in\N.
%    \end{array}
%\end{equation}
\end{dingli}
\begin{proof} {The basic idea of proving the strong NP-hardness of problem \eqref{problemmin} is to establish a polynomial time reduction from the MAX-2UNANIMITY problem to it. Since this proof is similar to the one of Theorem \ref{thm_complexity}, we just give the proof outline.} Given any instance of the MAX-2UNANIMITY problem with clauses $c_1,c_2,\ldots,c_m$ defined over Boolean variables $x_1,x_2,\ldots,x_n$
and an integer $M,$ we construct below a multi-user multi-carrier interference channel with $2n+2m+1$ users and $2$ subcarriers. In addition to ``variable user'' $i$ and ``auxiliary variable user'' $n+i$ for $i=1,2,\ldots,n$ and ``clause user'' $2n+j$ for $j=1,2,\ldots,m,$ we also construct ``auxiliary clause user'' $2n+m+j$ for $j=1,2,\ldots,m$ and ``super user'' $2n+2m+1.$ Hence, ${\K}=\{1,2,\ldots,2n+2m+1\}$ and $\N=\left\{1,2\right\}.$

{The difference between the above constructed setup and the one constructed in Theorem \ref{thm_complexity} is $m$ ``auxiliary clause users'' and $1$ ``super user''. The purpose of constructing these $m+1$ users is to establish the reduction from the MAX-2UNANIMITY problem to problem \eqref{problemmin}. More specifically, we construct channel parameters of all ``auxiliary clause users'' such that they and those of ``clause users'' are symmetric, i.e., if ``clause user'' $2n+j$ suffers interferences from ``variable user'' $i_1$ (``auxiliary variable user'' $n+i_1$) on subcarrier $1$ and ``auxiliary variable user'' $n+i_2$ (``variable user'' $i_2$) on subcarrier $2,$ then ``auxiliary clause user'' $2n+m+j$ suffers interferences from ``auxiliary variable user'' $n+i_1$ (``variable user'' $i_1$) on subcarrier $1$ and ``variable user'' $i_2$ (``auxiliary variable user'' $n+i_2$) on subcarrier $2.$ Then, we construct channel parameters of the ``super user'' such that all ``clause users'' and all ``auxiliary clause users'' cause equal interferences to the ``super user'' on two subcarriers. In this way, since channel parameters of all ``clause users'' and ``auxiliary clause users'' are symmetric, the ``super user'' can achieve a higher transmission rate if and only if ``clause users'' transmit less total power and cause less total interferences to it. Furthermore, it follows from the proof of Theorem \ref{thm_complexity} that ``clause users'' transmit less total power if and only if more clauses are satisfied unanimously in the given MAX-2UNANIMITY instance. This finishes the desirable reduction.}
%larger minimum rate is achieve by the
%less total transmission power is needed if and only if more clauses are satisfied unanimously.

{We construct a special interference channel with $2n+m$ users and $2$ subcarriers such that transmission rate expressions of all users take the following forms}:
%We now construct the direct-link and crosstalk channel gains among these $2n+m$ users on $2$ subcarriers.
for $i=1,2,\ldots,n,$ $R_{i}$ and $R_{n+i}$ are the same as the ones in \eqref{variablerate1} and \eqref{variablerate2}, respectively; for $j=1,2,\ldots,m,$ $R_{2n+j}$ is the same as the one in \eqref{clauserate} and %$R_{2n+m+j}$ as follows
%{\begin{equation}\label{clauserate-min1}R_{2n+2j-1}=\left\{\!\!\!\!\!\!
%  \begin{array}{ll}
%&\displaystyle \ln\left(1+\frac{p_{2n+2j-1}^1}{1+p_{2i_1-1}^1+p_{2i_2-1}^1}\right)+\ln\left(1+\frac{p_{2n+2j-1}^2}{1+p_{2i_1-1}^2+p_{2i_2-1}^2}\right),~\text{if}~c_j=x_{i_1}\vee x_{i_2};\\[15pt]
%&\displaystyle \ln\left(1+\frac{p_{2n+2j-1}^1}{1+p_{2i_1-1}^1+p_{2i_2}^1}\right)+\ln\left(1+\frac{p_{2n+2j-1}^2}{1+p_{2i_1-1}^2+p_{2i_2}^2}\right),~\text{if}~c_j=x_{i_1}\vee \bar x_{i_2};\\[15pt]
%&\displaystyle \ln\left(1+\frac{p_{2n+2j-1}^1}{1+p_{2i_1}^1+p_{2i_2-1}^1}\right)+\ln\left(1+\frac{p_{2n+2j-1}^2}{1+p_{2i_1}^2+p_{2i_2-1}^2}\right),~\text{if}~c_j=\bar x_{i_1}\vee x_{i_2};\\[15pt]
%&\displaystyle \ln\left(1+\frac{p_{2n+2j-1}^1}{1+p_{2i_1}^1+p_{2i_2}^1}\right)+\ln\left(1+\frac{p_{2n+2j-1}^2}{1+p_{2i_1}^2+p_{2i_2}^2}\right),~\text{if}~c_j=\bar x_{i_1}\vee \bar x_{i_2};
%  \end{array}\right.
%\end{equation}} and
{\begin{equation}\label{clauserate-min2}\!\!\!\!\!\!\!\!R_{2n+m+j}=\left\{\!\!\!\!\!\!\!\!
  \begin{array}{ll}
&\displaystyle \ln\left(1+\frac{p_{2n+m+j}^1}{1+p_{n+i_1}^1+p_{n+i_2}^1}\right)+\ln\left(1+\frac{p_{2n+m+j}^2}{1+p_{n+i_1}^2+p_{n+i_2}^2}\right),\,\text{if}\,c_j=x_{i_1}\vee x_{i_2};\\[15pt]
&\displaystyle \ln\left(1+\frac{p_{2n+m+j}^1}{1+p_{n+i_1}^1+p_{i_2}^1}\right)+\ln\left(1+\frac{p_{2n+m+j}^2}{1+p_{n+i_1}^2+p_{i_2}^2}\right),\,\text{if}\,c_j=x_{i_1}\vee \bar x_{i_2};\\[15pt]
&\displaystyle \ln\left(1+\frac{p_{2n+m+j}^1}{1+p_{i_1}^1+p_{n+i_2}^1}\right)+\ln\left(1+\frac{p_{2n+m+j}^2}{1+p_{i_1}^2+p_{n+i_2}^2}\right),\,\text{if}\,c_j=\bar x_{i_1}\vee x_{i_2};\\[15pt]
&\displaystyle \ln\left(1+\frac{p_{2n+m+j}^1}{1+p_{i_1}^1+p_{i_2}^1}\right)+\ln\left(1+\frac{p_{2n+m+j}^2}{1+p_{i_1}^2+p_{i_2}^2}\right),\,\text{if}\,c_j=\bar x_{i_1}\vee \bar x_{i_2};
  \end{array}\right.
\end{equation}}
and \begin{equation}\label{superclause}
  R_{2n+2m+1}=\ln\left(1+\frac{p_{2n+2m+1}^1}{1+\displaystyle\sum_{j=1}^{2m} p_{2n+j}^1}\right)+\ln\left(1+\frac{p_{2n+2m+1}^2}{1+\displaystyle\sum_{j=1}^{2m} p_{2n+j}^2}\right).
\end{equation} Moreover, let
\begin{equation}\label{budget}\bar p_{k}=\left\{\!\!\!\!\!\!
  \begin{array}{ll}
&\displaystyle 1,~\text{if}~1\leq k\leq 2n;\\[5pt]
&\displaystyle 4\left(\sqrt{2}-1\right),~\text{if}~2n+1\leq k\leq 2n+2m;\\[5pt]
&\displaystyle 2\left(\sqrt{2}-1\right)\left(1+2n+M+4(\sqrt{2}-1)(m-M)\right),~\text{if}~k=2n+2m+1.
  \end{array}\right.
\end{equation}
Then, the constructed instance of problem \eqref{problemmin} is
\begin{equation}\label{constructmin}
 \begin{array}{cl}
\displaystyle \max_{\left\{p_k^n\right\}} & \displaystyle \min_{k\in\K}\left\{R_{k}\right\}\\[5pt]%+\lambda\sum_{g=1}^{G}\|\left(\sum_{k=1}^K\bp_k\right)_{g}\|_{2} \\[20pt]
\text{s.t.} & \displaystyle p_k^1+p_k^2\leq 1,~k=1,2,\ldots,2n, \\[3pt]
            & \displaystyle p_k^1+p_k^2\leq 4\left(\sqrt{2}-1\right),~k=2n+1,2n+2,\ldots,2n+2m,\\[1pt]
            & p_{2n+2m+1}^1+p_{2n+2m+1}^2\leq 2\left(\sqrt{2}-1\right)\left(1+2n+M+4(\sqrt{2}-1)(m-M)\right),\\
        & p_k^1\geq 0,~p_k^2\geq 0,~k\in\K,
    \end{array}
\end{equation}where $R_k$ are given in \eqref{variablerate1}, \eqref{variablerate2}, \eqref{clauserate}, \eqref{clauserate-min2}, and \eqref{superclause}. {By using the similar argument as the one in Theorem \ref{thm_complexity}, we can show that the transformation from the MAX-2UNANIMITY problem to problem \eqref{constructmin} can be finished in polynomial time.} Next, we show that there exists a truth assignment such that at least $M$
clauses are satisfied unanimously for the given MAX-2UNANIMITY
instance if and only if the optimal value of problem \eqref{constructmin} is greater than or equal to $1.$

%In the above, each user $k$ is associated with two variables $p_k^1$ and $p_k^2$ for $k\in\K.$ The noise power of all users on all subcarriers are $1.$ For each ``variable user'' $2i-1,$ it suffers interference from ``auxiliary variable user'' $2i$ on both subcarriers $1$ and $2;$ each ``auxiliary variable user'' $2i$ suffers interference from ``variable user'' $2i-1$ on both subcarriers $1$ and $2;$ each ``clause user'' $2n+j$ suffers interference from ``variable user'' $2i_1-1$ and $2i_2-1$ and ``auxiliary variable user'' $2i_1$ and $2i_2,$ where $c_j$ contains literals of $x_{i_1}$ and $x_{i_2}.$ To make the construction of rate expressions clear, an illustrative
%example is given in Appendix \ref{app-example}.
%=\alpha_{i_1}\vee \beta_{i_2}$ and $\alpha,\beta\in\left\{x,\bar x\right\}.$

%[[[\textbf{Talk about communication background!}]]]
%
%Moreover, let $\bar p_k=1$ for all $k=1,2,\ldots,2n;$ $\bar p_k=4\left(\sqrt{2}-1\right)$ for all $k=2n+1,\ldots,2n+2m;$ and $\bar p_{2n+2m+1}=2\left(\sqrt{2}-1\right)\left(1+2n+M+4(m-M)(\sqrt{2}-1)\right).$

For any $j=1,2,\ldots,m,$ since $R_{2n+j}$ in \eqref{clauserate} and $R_{2n+m+j}$ in \eqref{clauserate-min2} are symmetric, it follows that $p_{2n+j}^1=p_{2n+m+j}^2$ and $p_{2n+m+j}^1=p_{2n+j}^2$ at the optimal solutions of problem \eqref{constructmin}. This shows that $$\sum_{j=1}^{2m} p_{2n+j}^1=\sum_{j=1}^{2m} p_{2n+j}^2=\sum_{j=1}^{m} p_{2n+j}^1+\sum_{j=1}^{m} p_{2n+j}^2$$ holds at the optimal solutions of problem \eqref{constructmin}, which, together with the fact $$\bar p_{2n+2m+1}=2\left(\sqrt{2}-1\right)\left(1+2n+M+4(\sqrt{2}-1)(m-M)\right),$$ implies that the optimal value of problem \eqref{constructmin} is greater than or equal to $1$ if and only if $$\displaystyle\sum_{j=1}^{m} p_{2n+j}^1+\sum_{j=1}^{m} p_{2n+j}^2$$ at the optimal solution is less than or equal to $2n+M+4(\sqrt{2}-1)(m-M).$ Furthermore, since the latter is true if and only if there exists a truth assignment such that $M$ clauses in the given
MAX-2UNANIMITY instance are unanimous (cf. Theorem \ref{thm_complexity}), we conclude that checking the optimal value of problem \eqref{constructmin} is greater than or equal to $1$ is strongly NP-hard. Therefore, problem \eqref{problemmin} is strongly NP-hard.
%the proof Theorem \ref{thm_complexity} and
\end{proof}

{Three remarks are in order. First, although Theorems \ref{thm_complexity} and \ref{thm-min} concentrate on the strong NP-hardness of problems \eqref{problem} and \eqref{problemmin} where $N=2,$ it is simple to use the same arguments to show the strong NP-hardness of more general problems \eqref{problem} and \eqref{problemmin} where $N\geq 2.$ Table \ref{table} summarizes the complexity status of spectrum management problems \eqref{problem} and \eqref{problemmin}. Second, although this paper focuses on spectrum management problems \eqref{problem} and \eqref{problemmin}, the developed techniques can be applied to show the NP-hardness of other related optimization problems {involving ``two''} arising from signal processing and wireless communications. For instance, the similar techniques have been used in \cite{liu11tspbeamforming} to show the NP-hardness of the harmonic-mean maximization problem under individual power constraints in the multi-user single-subcarrier multi-input single-output interference channel where each transmitter is equipped with two (or more) antennas.} {In Section \ref{sec-extension}, we shall also extend the techniques to show the strong NP-hardness of the linear transceiver design problem in the multi-user single-subcarrier multi-input multi-output (MIMO) interference channel where all transmitters and receivers are equipped with two (or more) antennas.} {Finally, we were drawn attention to the work \cite{locatelli} after the submission of this paper. In \cite{locatelli}, the authors showed that problem \eqref{problemmin} is NP-hard when $N=2$ based on a polynomial time reduction from the partition problem \cite{garey79complexity}. In contrast, we show in this paper that problem \eqref{problemmin} is strongly NP-hard when $N=2$ based on a polynomial time reduction from the MAX-2UNANIMITY~problem. Therefore, our proof technique is different from the one in \cite{locatelli} and our result is more stronger.}

%Finally, we were drawn attention to the related paper \cite{locatelli} after the submission of this paper. The paper \cite{locatelli} studies the complexity of the dynamic spectrum management under the harmonic mean utility and the proportional fairness utility and show that both of the problems are NP-hard.

% and therefore complement well with his paper. The min rate maximization problem is also shown to be NP-hard. However, Theorem
% the strong NP-hardness proofs of problems \eqref{problem} and \eqref{problemmin}  Theorems \ref{thm_complexity} and \ref{thm-min} can be easily extended to the general case where $N\geq 3$

\begin{table*}[!ht]
%\makeatletter\def\@captype{table}\makeatother
%\tabcolsep 5mm
%\renewcommand{\arraystretch}{1.3}
\caption{\textsc{Complexity Status of Spectrum Management for Multi-User Multi-Subcarrier Communication Systems}} \label{table} \centering
\begin{tabular}[h]{|c|c|c|}
\hline
  \hline
        {\backslashbox{\# of Subcarriers}{Problem}} &
{Total Power Minimization Problem \eqref{problem}} &  {Min-Rate Maximization Problem \eqref{problemmin}}
\\\hline
        {$N=1$}  &   {Polynomial\ Time Solvable \cite{luo08jstpdynamic}}        &     {Polynomial\ Time Solvable \cite{luo08jstpdynamic}}

\\\hline
        {$N\geq2$}    &        {Strongly NP-hard} (Theorem \ref{thm_complexity})   &    {Strongly NP-hard} (Theorem \ref{thm-min})    \\\hline
        \hline
\end{tabular}%\label{Xingzhi}
\end{table*}
%[[[\textbf{summarize a table!}]]]

%Remark: compare with Tom's result!

%\textbf{Questions:} Is it possible to extend the above complexity results to other utility maximization problems such as harmonic mean utility maximization?
%; existing algorithms for this problem; Design new efficient algorithms!

%\begin{itemize}
%  \item discuss $\xi^n$ not necessarily updated at each iteration;
%  \item projection and decoupleness
%  \item one update is enough!
%  \item BB algorithm
%  \item convergence
%\end{itemize}

{\subsection{Complexity Analysis of Linear Transceiver Design Problems}\label{sec-extension}

In this subsection, we first introduce two formulations of the linear transceiver design problem in the multi-user single-carrier MIMO interference channel and then apply our previously developed techniques to show that both of the problems are strongly NP-hard when all transmitters and receivers are equipped with two antennas.

%\subsection{Problem Formulation}\label{formulation-uv}
Consider a $K$-user single-carrier MIMO interference channel where the $k$-th
transmitter and receiver are equipped with $N_k$ and $M_k$ antennas,
respectively. The received signal at receiver $k$ is
%given as
\begin{equation}\nonumber
  \by_k=\bH_{k,k}\bv_ks_k+\sum_{j\neq k}\bH_{k,j}\bv_js_j+\bz_k,
\end{equation}
where $\bH_{k,j}\in\mathbb{C}^{M_k\times N_j}$ is the channel matrix
from transmitter $j$ to receiver $k$, $\bv_k\in\mathbb{C}^{N_k\times
1}$ is the beamformer used by transmitter $k$, $s_k\in\mathbb{C}$ is
the symbol that transmitter $k$ wishes to send to receiver $k$, and
$\bz_k\in\mathbb{C}^{M_k\times 1}$ is the additive white Gaussian
noise (AWGN) with distribution $\cal{CN}$$(\mathbf{0},
\eta_k\bI).$ Each receiver uses a linear receive strategy and let
$\bu_k\in\mathbb{C}^{M_k\times 1}$ be the receive beamformer of receiver $k$.
Then, the linearly processed signal at the $k$-th receiver is
\begin{equation}\nonumber
  \hat s_k=\bu_k^\dagger\by_k,
\end{equation} where $\left(\cdot\right)^{\dagger}$ denotes the Hermitian operator.
Treating interference as noise, we can write the SINR of user $k$ as
\begin{equation*}\label{SINR}
\displaystyle
\SI_k=\frac{|\bu_k^\dagger\bH_{k,k}\bv_k|^2}{\eta_k\|\bu_k\|^2+\displaystyle\sum_{j\neq
k}|\bu_k^\dagger\bH_{k,j}\bv_j|^2}.
\end{equation*} We consider the following two formulations of the linear transceiver design problem:
\begin{equation}\label{min-power-uv}
\begin{array}{cl}
\displaystyle \max_{\left\{\bu_k,\,\bv_k\right\}} &\displaystyle \sum_{k\in\cal K} \|\bv_k\|^2 \\
[5pt] \mbox{s.t.} & \SI_k\geq \gamma_k,~\|\bu_k\|^2=1,\,k\in \cal K,
\end{array}
\end{equation}
and
\begin{equation}\label{max-min-uv}
\begin{array}{cl}
\displaystyle \max_{\left\{\bu_k,\,\bv_k\right\}} &\displaystyle \min_{k\in\cal K}\left\{\SI_k\right\} \\
[5pt] \mbox{s.t.} &
\|\bu_k\|^2=1,~\|\bv_k\|^2\le \bar p_k,\,k\in \cal K,
\end{array}
\end{equation}where $\gamma_k$ is the desired SINR target of user $k$ and $\bar p_k$ is the power budget of transmitter $k.$ Notice that the norm of all receive beamformers is normalized to be one in problems \eqref{min-power-uv} and \eqref{max-min-uv}.

%\subsection{Strong NP-Hardness of Problems \eqref{min-power-uv} and \eqref{max-min-uv} with $M_k=N_k=2$} \label{subsec:newcom-uv}
%In this subsection, we prove that both problems \eqref{min-power-uv} and \eqref{max-min-uv} are strongly NP-hard when $M_k=N_k=2$ for all $k\in\K.$

%We now present the main results in this subsection.
We have the following strong NP-hardness results.

\begin{dingli}\label{thm-min-power-uv}
Problem \eqref{min-power-uv} is strongly NP-hard when $M_k=N_k=2$ for all $k\in\K.$
\end{dingli}

The detailed proof of Theorem \ref{thm-min-power-uv} can be found in Appendix \ref{app:lemma-uv}. By using the similar argument in the proof of Theorem \ref{thm-min-power-uv} and the similar technique as in the proof of Theorem \ref{thm-min}, we can show the following Theorem \ref{thm-max-min-uv}. We leave the proof of Theorem \ref{thm-max-min-uv} as an exercise for the interested readers. %XXXXXXX
% strong NP-hardness result.

\begin{dingli}\label{thm-max-min-uv}
Problem \eqref{max-min-uv} is strongly NP-hard when $M_k=N_k=2$ for all $k\in\K.$
\end{dingli}
}

{Notice that the same strong NP-hardness results in Theorems \ref{thm-min-power-uv} and \ref{thm-max-min-uv} hold true for more general problems \eqref{min-power-uv} and \eqref{max-min-uv} where $\min\left\{M_k,N_k\right\}\geq 2$ for all $\K,$ although Theorems \ref{thm-min-power-uv} and \ref{thm-max-min-uv} concentrate on the special case where $M_k=N_k=2$ for all $\K.$ Table \ref{table-uv} summarizes the complexity status of linear transceiver design problems \eqref{min-power-uv} and \eqref{max-min-uv}. %All the results in Table \ref{table-uv} also hold true for total power minimization problem \eqref{min-power-uv}.
Next, we give some remarks on the complexity of max-min fairness linear transceiver design problem \eqref{max-min-uv}.}

%It is worthwhile remarking that problem \eqref{max-min-uv} has been shown to be NP-hard in \cite{liu2013maxsimo,liu2011maxicc,meisam2013linear}. More specifically, it is shown

{Problem \eqref{max-min-uv} is shown to be strongly NP-hard in \cite{liu2011maxicc,liu2013max} when $\min\left\{M_k,N_k\right\}\geq 2$ and $M_k+N_k\geq 5$ for all $k\in\K.$ The proof in \cite{liu2011maxicc,liu2013max} is based on a polynomial time reduction from the 3-SATISFIABILITY problem. Then, (a variant of) problem \eqref{max-min-uv} is shown to remain strongly NP-hard in \cite{meisam2013linear} when $\min\left\{M_k,N_k\right\}\geq 2$ for all $k\in\K.$ The proof in \cite{meisam2013linear} is based on a polynomial time reduction from the same NP-hard problem as the one used in \cite{liu2013max}. In this paper, we show the strong NP-hardness of problem \eqref{max-min-uv} when $\min\left\{M_k,N_k\right\}\geq 2$ for all $k\in\K$ by establishing a polynomial time reduction from the MAX2-UNANIMITY problem, which is different from the ones in \cite{liu2011maxicc,liu2013max,meisam2013linear}. Moreover, our NP-hardness result is stronger than the one in \cite{meisam2013linear}, since the result in \cite{meisam2013linear} holds true only for complex channel matrices while our result holds true for both complex and real channel matrices.} %Therefore, this paper gives a complete complexity analysis of the max-min fairness linear transceiver design problem.

\begin{table*}[!ht]
%\makeatletter\def\@captype{table}\makeatother
%\tabcolsep 5mm
%\renewcommand{\arraystretch}{1.3}
\caption{\textsc{Complexity Status of Linear
Transceiver Design Problems \eqref{min-power-uv} and \eqref{max-min-uv}}} \label{table-uv} \centering
\begin{tabular}[h]{|c|c|c|}
\hline
  \hline
        {\backslashbox{\# of Rx Antennas}{\# of Tx Antennas}} &
{$N_k=1$} &  {$N_k\geq2$}
\\\hline
        {$M_k=1$}  &   {Polynomial\ Time Solvable} \cite{luo08jstpdynamic}       &     {Polynomial\ Time Solvable} \cite{wiesel2006linear,liu11tspbeamforming}
        \\\hline
        {$M_k\geq2$}  &   {Polynomial\ Time Solvable} \cite{liu2013maxsimo}       &     {Strongly NP-hard}  (\!\!\cite{meisam2013linear}, Theorems \ref{thm-min-power-uv} and \ref{thm-max-min-uv})
\\\hline
        %{$M_k\geq3$}    &        {Poly.\ Time Solvable} \cite{liu2013maxsimo}   &    {Strongly NP-hard} (Theorems \ref{thm-min-power-uv} and \ref{thm-max-min-uv})    &      {Strongly NP-hard} (Theorems \ref{thm-min-power-uv} and \ref{thm-max-min-uv})  \\\hline
%        \hline
\end{tabular}%\label{Xingzhi}
\end{table*}

%\begin{table*}[!ht]
%%\makeatletter\def\@captype{table}\makeatother
%%\tabcolsep 5mm
%%\renewcommand{\arraystretch}{1.3}
%\caption{\textsc{Complexity Status of Linear
%Transceiver Design Problems \eqref{min-power-uv} and \eqref{max-min-uv}}} \label{table-uv} \centering
%\begin{tabular}[h]{|c|c|c|c|}
%\hline
%  \hline
%        {\backslashbox{Rx}{Tx}} &
%{$N_k=1$} &  {$N_k=2$} &  {$N_k\geq3$}
%\\\hline
%        {$M_k=1$}  &   {Poly.\ Time Solvable} \cite{luo08jstpdynamic}       &     {Poly.\ Time Solvable} \cite{wiesel2006linear,liu11tspbeamforming}         &
%        {Poly.\ Time Solvable} \cite{wiesel2006linear,liu11tspbeamforming}
%        \\\hline
%        {$M_k=2$}  &   {Poly.\ Time Solvable} \cite{liu2013maxsimo}       &     {Strongly NP-hard}  (\!\!\cite{meisam2013linear}, Theorems \ref{thm-min-power-uv} and \ref{thm-max-min-uv})        &
%        {Strongly NP-hard} (Theorems \ref{thm-min-power-uv} and \ref{thm-max-min-uv})
%\\\hline
%        {$M_k\geq3$}    &        {Poly.\ Time Solvable} \cite{liu2013maxsimo}   &    {Strongly NP-hard} (Theorems \ref{thm-min-power-uv} and \ref{thm-max-min-uv})    &      {Strongly NP-hard} (Theorems \ref{thm-min-power-uv} and \ref{thm-max-min-uv})  \\\hline
%        \hline
%\end{tabular}%\label{Xingzhi}
%\end{table*}
%
%\cite{liu2013maxsimo}
%
%\cite{meisam2013linear}
%
%\cite{wiesel2006linear}

%Remark: Difference between the existing works!
%Real Case. Complete Analysis

\section{Conclusions}
Dynamic spectrum management in accordance with fast channel fluctuations can significantly improve spectral efficiency of the multi-user multi-carrier communication system. A major challenge associated with spectrum management is to find, for a given channel state, the globally optimal spectrum management strategy to minimize the total transmission power or maximize the system utility. This paper has provided a complete complexity characterization of the spectrum management problem in the multi-user multi-subcarrier communication system. We have shown that both the total power minimization problem and the min-rate maximization problem are strongly NP-hard when the number of subcarriers is two, and thus answered a long-standing open question in the literature. The complexity results suggest that there is no polynomial time algorithms which can solve the general spectrum management problem to global optimality (unless P$=$NP) and {it is more realistic to design efficient algorithms for finding an approximately optimal or locally optimal spectrum management strategy in polynomial time in practice.} {It is worthwhile pointing out that the developed techniques in this paper can potentially be extended to show the (strong) NP-hardness of other related optimization problems {involving ``two'' arising from} signal processing and wireless communications.} {Our future work is to design efficient approximation algorithms (with guaranteed approximation ratios) for solving the general dynamic spectrum management problem.}

%, albeit this paper focuses on the dynamic spectrum management problem.

%\rev{Based on our knowledge, there is no approximation algorithms (with guaranteed approximation ratios) for solving the general dynamic spectrum management problem in the literatures. We would like to take the design of approximation algorithms for the general spectrum management problem as our future work.}

%We have answered an open question by showing that the min-rate maximization problem is strongly NP-hard when the number of subcarriers is two. These %complexity
%results reveal that the power allocation problem in the multi-user multi-subcarrier system is intrinsically difficult to solve (except some special
%cases) and therefore provide valuable information to algorithm designers in directing their efforts toward those approaches
%that have the greatest potential of leading to useful algorithms.

\section*{Acknowledgment}%{acknowledgements}
%We would like to thank Prof. Xiaojun Chen and Dr. Wei Bian for many insightful comments, which helped us in improving the results in this paper.
 The author thanks Professor Zhi-Quan (Tom) Luo of University of Minnesota and The Chinese University of Hong Kong (Shenzhen) and Professor Yu-Hong Dai of Chinese Academy of Sciences for many useful discussions on an early version of this paper. {The author also thanks the Associate Editor, Professor Osvaldo Simeone, and three anonymous reviewers for their constructive comments, which significantly improved the quality and presentation of the paper.}

\appendices

\section{A Lemma}\label{app:lemma}%
%\begin{proof}
\begin{yinli}\label{lemma_small_power}
  The points $\left(p_1^1,p_1^2,p_2^1,p_2^2\right)^T=(1,0,0,1)^T$ and $\left(p_1^1,p_1^2,p_2^1,p_2^2\right)^T=(0,1,1,0)^T$ are the only feasible solutions of \begin{equation}\label{small_problem_feasibility}
 \left\{\!\!\!\!\!\!\begin{array}{cl}
 & \ln\left(1+\frac{p_1^1}{1+p_2^1}\right)+\ln\left(1+\frac{p_1^2}{1+p_2^2}\right)\geq \ln 2, \\[10pt]
           & \ln\left(1+\frac{p_2^1}{1+p_1^1}\right)+\ln\left(1+\frac{p_2^2}{1+p_1^2}\right)\geq \ln 2, \\[10pt]
          & p_1^1+p_1^2\leq 1,~p_2^1+p_2^2\leq 1, \\[3pt]
        & p_1^1\geq0,~p_1^2\geq 0,~p_2^1\geq 0,~p_2^2\geq 0.
    \end{array}\right.
\end{equation}
  %are optimal solutions of problem \eqref{small_problem}. Moreover, they are the only two feasible points of problem \eqref{small_problem_feasibility}. %the following problem
%    \begin{equation}\label{small_feasible_problem}\left\{
% \begin{array}{cl}
% & \ln\left(1+\frac{p_1^1}{1+p_2^1}\right)+\ln\left(1+\frac{p_1^2}{1+p_2^2}\right)\geq \ln 2, \\[10pt]
%           & \ln\left(1+\frac{p_2^1}{1+p_1^1}\right)+\ln\left(1+\frac{p_2^2}{1+p_1^2}\right)\geq \ln 2, \\[5pt]
%           & p_1^1+p_1^2\leq 2,~p_2^1+p_2^2\leq 2, \\[5pt]
%        & p_1^1\geq0,~p_1^2\geq 0,~p_2^1\geq 0,~p_2^2\geq 0.
%    \end{array}\right.
%\end{equation}
\end{yinli}
%\begin{proof}
 \emph{Proof of Lemma \ref{lemma_small_power}:} We first prove that $\left(p_1^1,p_1^2,p_2^1,p_2^2\right)^T=(1,0,0,1)^T$ and $\left(p_1^1,p_1^2,p_2^1,p_2^2\right)^T=(0,1,1,0)^T$ are the only two optimal solutions of problem
\begin{equation}\label{small_problem}
 \begin{array}{cl}
\displaystyle \min_{p_1^1,\,p_1^2,\,p_2^1,\,p_2^2} & \displaystyle p_1^1+p_1^2+p_2^1+p_2^2\\[10pt]%+\lambda\sum_{g=1}^{G}\|\left(\sum_{k=1}^K\bp_k\right)_{g}\|_{2} \\[20pt]
\text{s.t.} & \ln\left(1+\frac{p_1^1}{1+p_2^1}\right)+\ln\left(1+\frac{p_1^2}{1+p_2^2}\right)\geq \ln 2, \\[10pt]
           & \ln\left(1+\frac{p_2^1}{1+p_1^1}\right)+\ln\left(1+\frac{p_2^2}{1+p_1^2}\right)\geq \ln 2, \\[5pt]
        & p_1^1\geq0,~p_1^2\geq 0,~p_2^1\geq 0,~p_2^2\geq 0.
    \end{array}
\end{equation} The two rate constraints in problem \eqref{small_problem} can be equivalently rewritten as
  $$\left(1+p_1^1+p_2^1\right)\left(1+p_1^2+p_2^2\right)\geq 2\left(1+p_2^1\right)\left(1+p_2^2\right)$$
  and $$\left(1+p_1^1+p_2^1\right)\left(1+p_1^2+p_2^2\right)\geq 2\left(1+p_1^1\right)\left(1+p_1^2\right).$$ Adding the above two inequalities together yields $p_1^1p_2^2+p_1^2p_2^1\geq 1,$ which implies that problem
  \begin{equation}\label{relax_small_problem}
 \begin{array}{cl}
\displaystyle \min_{p_1^1,\,p_1^2,\,p_2^1,\,p_2^2} & \displaystyle p_1^1+p_1^2+p_2^1+p_2^2\\
\text{s.t.} & p_1^1p_2^2+p_1^2p_2^1\geq 1, \\%[10pt]
        & p_1^1\geq0,~p_1^2\geq 0,~p_2^1\geq 0,~p_2^2\geq 0,
    \end{array}
\end{equation} is a relaxation of problem \eqref{small_problem}. If we can show $\left(p_1^1,p_1^2,p_2^1,p_2^2\right)^T=(1,0,0,1)^T$ and $\left(p_1^1,p_1^2,p_2^1,p_2^2\right)^T=(0,1,1,0)^T$ are the only two optimal solutions of problem \eqref{relax_small_problem}, then they must be the only two optimal solutions of problem \eqref{small_problem}. This further implies that $\left(p_1^1,p_1^2,p_2^1,p_2^2\right)^T=(1,0,0,1)^T$ and $\left(p_1^1,p_1^2,p_2^1,p_2^2\right)^T=(0,1,1,0)^T$ are the only feasible points of
\begin{equation}\label{small_problem_feasibility2}
 \left\{\begin{array}{cl}
 & \ln\left(1+\frac{p_1^1}{1+p_2^1}\right)+\ln\left(1+\frac{p_1^2}{1+p_2^2}\right)\geq \ln 2, \\[10pt]
           & \ln\left(1+\frac{p_2^1}{1+p_1^1}\right)+\ln\left(1+\frac{p_2^2}{1+p_1^2}\right)\geq \ln 2, \\[10pt]
          & p_1^1+p_1^2+p_2^1+p_2^2\leq 2, \\[3pt]
        & p_1^1\geq0,~p_1^2\geq 0,~p_2^1\geq 0,~p_2^2\geq 0.
    \end{array}\right.
\end{equation} Since $$\left\{\left(p_1^1,p_1^2,p_2^1,p_2^2\right)\geq 0\,|\,p_1^1+p_1^2\leq 1,\,p_2^1+p_2^2\leq 1\right\}\subseteq\left\{\left(p_1^1,p_1^2,p_2^1,p_2^2\right)\geq 0\,|\,p_1^1+p_1^2+p_2^1+p_2^2\leq 2\right\},$$ it follows that $\left(p_1^1,p_1^2,p_2^1,p_2^2\right)^T=(1,0,0,1)^T$ and $\left(p_1^1,p_1^2,p_2^1,p_2^2\right)^T=(0,1,1,0)^T$ are the only feasible points of \eqref{small_problem_feasibility}.

It remains to prove that $\left(p_1^1,p_1^2,p_2^1,p_2^2\right)^T=(1,0,0,1)^T$ and $\left(p_1^1,p_1^2,p_2^1,p_2^2\right)^T=(0,1,1,0)^T$ are the only two optimal solutions of problem \eqref{relax_small_problem}. %Next, we show this indeed is true.
%If we can show that the optimal solution to the above problem is feasible to problem \eqref{small_problem}, then they must be the optimal solutions to problem \eqref{small_problem}. Next, we show $\left(p_1^1,p_1^2,p_2^1,p_2^2\right)^T=(1,0,0,1)^T$ and $\left(p_1^1,p_1^2,p_2^1,p_2^2\right)^T=(0,1,1,0)^T$ are the only two solutions to problem \eqref{relax_small_problem}.
It can be verified that the optimal solution of the following problem  \begin{equation*}
 \begin{array}{cl}
\displaystyle \min_{p_1^1,\,p_1^2,\,p_2^1,\,p_2^2} & \displaystyle 2\left(\sqrt{p_1^1p_2^2}+\sqrt{p_1^2p_2^1}\right)\\%[15pt]
\text{s.t.} & p_1^1p_2^2+p_1^2p_2^1\geq 1, \\%[10pt]
        & p_1^1\geq0,~p_1^2\geq 0,~p_2^1\geq 0,~p_2^2\geq 0,
    \end{array}
\end{equation*} must satisfy $$p_1^1p_2^2=1,~p_1^2p_2^1=0$$ or $$p_1^1p_2^2=0,~p_1^2p_2^1=1,$$ and its optimal value is $2.$ Since $p_1^1+p_2^2\geq 2\sqrt{p_1^1p_2^2}$ and $p_1^2+p_2^1\geq 2\sqrt{p_1^2p_2^1}$ and the above two inequalities hold true with ``='' if and only if $p_1^1=p_2^2$ and $p_1^2=p_2^1,$ we conclude that the only points that achieve the objective value of problem \eqref{relax_small_problem} of being $2$ are
$\left(p_1^1,p_1^2,p_2^1,p_2^2\right)^T=(1,0,0,1)^T$ and $\left(p_1^1,p_1^2,p_2^1,p_2^2\right)^T=(0,1,1,0)^T.$ Hence, $\left(p_1^1,p_1^2,p_2^1,p_2^2\right)^T=(1,0,0,1)^T$ and $\left(p_1^1,p_1^2,p_2^1,p_2^2\right)^T=(0,1,1,0)^T$ are the only two optimal solutions of problem \eqref{relax_small_problem}. This completes the proof of Lemma \ref{lemma_small_power}.% is completed.
%\end{proof}

\section{An Illustrative Example}\label{app-example}%
%To make the construction of rate expressions clear, an illustrative
%example is given, where
%In this example, there are four disjunctive clauses
%$c_{1}=x_1\vee {\bar x}_2,~c_2=x_1\vee {x}_3,$ $c_3=\bar x_2\vee\bar x_4,$ and $c_4=\bar x_3\vee x_4$ defined over
%four Boolean variables $x_1,x_2,x_3,$ and $x_4.$
Consider the instance in Example \ref{example}.
Then there
are $12$ users in the constructed 2-carrier communication system,
including $4$ ``variable user" (denoted as user $1,2,3,4$),~$4$ ``auxiliary variable user" (denoted as user $5,6,7,8$), and $4$
``clause user" (denoted as user $9,10,11,12$). In this case, ${\cal
K}=\left\{1,2,\ldots,12\right\},$ and all users' rate expressions are given as follows:
\begin{align*}
   & R_{1}=\ln\left(1+\frac{p_{1}^1}{1+p_{5}^1}\right)+\ln\left(1+\frac{p_{1}^2}{1+p_{5}^2}\right),\\[3pt]
   & R_{2}=\ln\left(1+\frac{p_{2}^1}{1+p_{6}^1}\right)+\ln\left(1+\frac{p_{2}^2}{1+p_{6}^2}\right),\\[3pt]
    & R_{3}=\ln\left(1+\frac{p_{3}^1}{1+p_{7}^1}\right)+\ln\left(1+\frac{p_{3}^2}{1+p_{7}^2}\right),\\[3pt]
    & R_{4}=\ln\left(1+\frac{p_{4}^1}{1+p_{8}^1}\right)+\ln\left(1+\frac{p_{4}^2}{1+p_{8}^2}\right),\\[3pt]
  & R_{5}=\ln\left(1+\frac{p_{5}^1}{1+p_{1}^1}\right)+\ln\left(1+\frac{p_{5}^2}{1+p_{1}^2}\right),\\[3pt]
 & R_{6}=\ln\left(1+\frac{p_{6}^1}{1+p_{2}^1}\right)+\ln\left(1+\frac{p_{6}^2}{1+p_{2}^2}\right),\\[3pt]
  & R_{7}=\ln\left(1+\frac{p_{7}^1}{1+p_{3}^1}\right)+\ln\left(1+\frac{p_{7}^2}{1+p_{3}^2}\right),\\[3pt]
  & R_{8}=\ln\left(1+\frac{p_{8}^1}{1+p_{4}^1}\right)+\ln\left(1+\frac{p_{8}^2}{1+p_{4}^2}\right),\\[3pt]
 & R_{9}=\ln\left(1+\frac{p_{9}^1}{1+p_{1}^1+p_{6}^1}\right)+\ln\left(1+\frac{p_{9}^2}{1+p_{1}^2+p_{6}^2}\right),\\[3pt]
 & R_{10}=\ln\left(1+\frac{p_{10}^1}{1+p_{1}^1+p_{3}^1}\right)+\ln\left(1+\frac{p_{10}^2}{1+p_{1}^2+p_{3}^2}\right),\\[3pt]
 & R_{11}=\ln\left(1+\frac{p_{11}^1}{1+p_{6}^1+p_{8}^1}\right)+\ln\left(1+\frac{p_{11}^2}{1+p_{6}^2+p_{8}^2}\right),\\[3pt]
 & R_{12}=\ln\left(1+\frac{p_{12}^1}{1+p_{7}^1+p_{4}^1}\right)+\ln\left(1+\frac{p_{12}^2}{1+p_{7}^2+p_{4}^2}\right).
\end{align*}

{\section{Proof of Theorem \ref{thm-min-power-uv}}\label{app:lemma-uv}

To ease the presentation, we first define some notation. Let
\begin{equation}\nonumber
  \begin{array}{cc} \be_1 = \left(
                                \begin{array}{c}
                                  1 \\
                                  0 \\
                                \end{array}
                              \right),~\be_2 = \left(
                                \begin{array}{c}
                                  0 \\
                                  1 \\
                                \end{array}
                              \right),~\be = \left(
                                \begin{array}{c}
                                  1 \\
                                  1 \\
                                \end{array}
                              \right),\\[20pt]\bH_A=\left(
                                               \begin{array}{cc}
                                                 1 & 0 \\
                                                 0 & 1 \\
                                               \end{array}
                                             \right),~\bH_B=\left(
                                               \begin{array}{cc}
                                                 0 & 1 \\
                                                 1 & 0 \\
                                               \end{array}
                                             \right),~\bH_C=\left(
                                               \begin{array}{cc}
                                                 0 & 1 \\
                                                 0 & 0 \\
                                               \end{array}
                                             \right).
  \end{array}
\end{equation}
For any given two vectors $\bu$ and $\bv$ with $\|\bu\|=\|\bv\|=1,$ we denote $\bu\simeq\bv$ if there exists a scaler $r\in\mathbb{C}$ (with $|r|=1$) such that $\bu=r\bv.$

To show Theorem \ref{thm-min-power-uv}, we need the following lemma, which recognizes a discrete structure in the solution of a special instance of the decision version of problem \eqref{min-power-uv}.
\begin{yinli}\label{lemma_small_power_uv}
Consider the following problem instance \begin{equation}\label{small_problem_feasibility_2}
 \left\{\!\!\!\!\!\!\begin{array}{cl}
 & \left|\bu^{\dagger}_1\bH_A\bv_1\right|^2\geq 1 + \left|\bu^{\dagger}_1\bH_B\bv_2\right|^2+ \left|\bu^{\dagger}_1\bH_B\bv_3\right|^2, \\[5pt]
 & \left|\bu^{\dagger}_2\bH_A\bv_2\right|^2\geq 1 + \left|\bu^{\dagger}_2\bH_C\bv_3\right|^2, \\[5pt]
 & \left|\bu^{\dagger}_3\bH_A\bv_3\right|^2\geq 1 + \left|\bu^{\dagger}_3\bH_C\bv_2\right|^2, \\[5pt]
 & \left\|\bv_1\right\|^2+\left\|\bv_2\right\|^2+\left\|\bv_3\right\|^2\leq 3,\\[5pt]
  & \left\|\bu_1\right\|^2=1,\,\left\|\bu_2\right\|^2=1,\,\left\|\bu_3\right\|^2 = 1.
    \end{array}\right.
\end{equation}
  The points $\bu_1\simeq\bu_2\simeq\bu_3\simeq\bv_1\simeq\bv_2\simeq\bv_3\simeq\be_1$ and $\bu_1\simeq\bu_2\simeq\bu_3\simeq\bv_1\simeq\bv_2\simeq\bv_3\simeq\be_2$ are the only two feasible solutions of problem \eqref{small_problem_feasibility_2} (up to an arbitrary phase rotation).
\end{yinli}

%\begin{proof} %To ease the presentation, we neglect the phase rotation in the proof.

\emph{Proof of Lemma \ref{lemma_small_power_uv}:} We first prove that the necessary conditions for problem \eqref{small_problem_feasibility_2} being feasible are
  \begin{equation}\label{eq1}\left\|\bv_1\right\|=\left\|\bv_2\right\|=\left\|\bv_3\right\| = 1,\end{equation}
  \begin{equation}\label{eq2}\bu_1\simeq\bv_1,~\bu_2\simeq\bv_2,~\bu_3\simeq\bv_3,\end{equation}
  and
  \begin{equation}\label{eq3}\bu^{\dagger}_1\bH_B\bv_2=0,~\bu^{\dagger}_1\bH_B\bv_3=0,~\bu^{\dagger}_2\bH_C\bv_3=0.\end{equation}
We prove \eqref{eq1}, \eqref{eq2}, and \eqref{eq3} sequentially.

{Proof of \eqref{eq1}:} We prove \eqref{eq1} based on the contradiction principle. We divide the proof into two cases: Case 1 (there exists an $i=1,2,3$ such that $\|\bv_i\|< 1$) and Case 2 (there exists an $i=1,2,3$ such that $\|\bv_i\|> 1$).
\begin{itemize}
  \item [-] Case 1: Without loss of generality, suppose that $\|\bv_1\|< 1$. This, together with the fact $\|\bu_1\|=1,$ implies $\left|\bu^{\dagger}_1\bH_A\bv_1\right|^2<1,$ which contradicts the first equation of \eqref{small_problem_feasibility_2}. Consequently, we must have $\|\bv_i\|\geq 1$ for all $i=1,2,3.$ Combining this with the fourth equation of \eqref{small_problem_feasibility_2}, we immediately obtain \eqref{eq1}.
  \item [-] Case 2: Without loss of generality, suppose that $\|\bv_1\|> 1.$ Then, by the fourth equation of \eqref{small_problem_feasibility_2}, we must have $\|\bv_2\|< 1$ or $\|\bv_3\|< 1$. This reduces to Case 1 and therefore \eqref{eq1} is true.
\end{itemize}

{Proof of \eqref{eq2}:} It follows from the first three conditions of \eqref{small_problem_feasibility_2} that $\left|\bu_i^{\dagger}\bv_i\right|\geq 1$ for all $i=1,2,3.$ This and the facts $\|\bu_i\|=\|\bv_i\|=1$ for all $i=1,2,3$ imply \eqref{eq2}. %$\bu_i\simeq\bv_i$ for all $i=1,2,3,$

{Proof of \eqref{eq3}:} By \eqref{eq1}, \eqref{eq2}, and the last condition of \eqref{small_problem_feasibility_2}, we have $\left|\bu^{\dagger}_i\bH_A\bv_i\right|=1$ for all $i=1,2,3.$ This, together with the first three conditions in \eqref{small_problem_feasibility_2}, immediately implies \eqref{eq3}.

%\begin{itemize}
%  \item We prove \eqref{eq1} based on the contradiction principle. Without loss of generality, suppose that $\|\bv_1\|< 1$ . This, together with the fact $\|\bu_1\|=1,$ implies $\left|\bu^{\dagger}_1\bH_A\bv_1\right|^2<1,$ which contradicts the first equation of \eqref{small_problem_feasibility_2}. Then, we must have $\|\bv_i\|\geq 1$ for all $i=1,2,3.$ Combining this with the fourth equation of \eqref{small_problem_feasibility_2}, we immediately obtain \eqref{eq1}.
%  \item It follows from the first three conditions in \eqref{small_problem_feasibility_2} that $\left|\bu_i^{\dagger}\bv_i\right|\geq 1$ for all $i=1,2,3.$ Combining this with the facts $\|\bu_i\|=\|\bv_i\|=1$ for all $i=1,2,3$ yields \eqref{eq2}. %$\bu_i\simeq\bv_i$ for all $i=1,2,3,$
%  \item By \eqref{eq1}, \eqref{eq2}, and the last condition in \eqref{small_problem_feasibility_2}, we have $\left|\bu^{\dagger}_i\bH_A\bv_i\right|=1$ for all $i=1,2,3.$ This, together with the first three conditions in \eqref{small_problem_feasibility_2}, immediately implies \eqref{eq3}.
%\end{itemize}

Next, we show the truth of the lemma based on \eqref{eq1}, \eqref{eq2}, and \eqref{eq3}. We focus on showing that $\bv_1\simeq\be_1$ or $\bv_1\simeq\be_2.$ Based on this, we can immediately obtain the other results in the lemma by using \eqref{eq2} and \eqref{eq3}. Let $$\bv_1=\left(
               \begin{array}{c}
                 v_{11} \\
                 v_{12} \\
               \end{array}
             \right)
  $$ with $\|\bv_{1}\|=1.$ %satisfying \eqref{alignment}.
  It suffices to prove that either $v_{11}=0$ or $v_{12}=0.$ Combining \eqref{eq2} and \eqref{eq3} yields
  \begin{equation}\label{alignment}\bv^{\dagger}_1\bH_B\bv_2=0,~\bv^{\dagger}_1\bH_B\bv_3=0,~\bv^{\dagger}_2\bH_C\bv_3=0.\end{equation}
  By the first two conditions of \eqref{alignment}, we obtain
  \begin{equation}\label{v2v3}\bv_2\simeq \bH_B \left(
                      \begin{array}{c}
                        \bar v_{12} \\
                        - \bar v_{11} \\
                      \end{array}
                    \right)~\text{and}~\bv_3\simeq \bH_B \left(
                      \begin{array}{c}
                        \bar v_{12} \\
                        - \bar v_{11} \\
                      \end{array}
                    \right),
  \end{equation} where $\bar v$ denotes the conjugate of $v\in\mathbb{C}.$ %where $r_2$ and $r_3$ are two complex numbers with $|r_1|=|r_2|=1.$
  Substituting \eqref{v2v3} into the third condition of \eqref{alignment}, we get
  $$\left(
                      \begin{array}{c}
                        \bar v_{12} \\
                        - \bar v_{11} \\
                      \end{array}
                    \right)^{\dagger} \bH_B \bH_C \bH_B\left(
                      \begin{array}{c}
                        \bar v_{12} \\
                        - \bar v_{11} \\
                      \end{array}
                    \right)=0,
  $$ which is equivalent to $$v_{11} \bar v_{12} = 0.$$
  This shows that either $v_{11}$ or $v_{12}$ is zero. Therefore, we get either $\bv_1\simeq\be_1$ or $\bv_1\simeq\be_2.$ The proof of Lemma \ref{lemma_small_power_uv} is completed.
  %$$- \bar r_1 r_2 v_{11} \bar v_{12} = 0.$$
  %\bv^{\dagger}_2\bH_C\bv_3 = \bar r_1 r_2
%\end{proof}

%\begin{IEEEproof} %The proof is based on a polynomial time reduction from the MAX-2UNANIMITY problem. Specifically, we claim that the following feasibility
%problem is strongly NP-hard, i.e.,
%  checking whether there exist beamforming vectors $\bu_k,\bv_k~(k\in\cal K)$ such that
%  all users' SINR levels are greater than or equal to the given SINR target
%  $\zeta:$
%\begin{equation}\nonumber\left\{\!\!\!\!\!
%\begin{array}{ll}
%&\displaystyle  \frac{|\bu_k^\dagger\bH_{kk}\bv_k|^2}{\sigma^2_k\|\bu_k\|^2+\displaystyle\sum_{j\neq
%k}|\bu_k^\dagger\bH_{kj}\bv_j|^2}\geq\zeta,\,k\in\cal K;\\[3pt]
%&\displaystyle \|\bu_k\|^2=1,\,k\in\cal K;\\[3pt]
%&\displaystyle \sum_{k\in\cal K}\|\bv_k\|^2\leq \bar P.
%\end{array}\right.
%\end{equation}
%In the above, the channel matrices are real-valued and $\bH_{kj}\in
%\mathbb R^{2\times2}, \forall~k,j\in \cal K.$ We remark that all the
%complexity results in this paper can extend to complex channel
%matrices.
We are now ready to prove Theorem \ref{thm-min-power-uv}.

\emph{Proof of Theorem \ref{thm-min-power-uv}:} Given any instance of the MAX-2UNANIMITY problem with clauses
$c_1,c_2,\ldots,c_m$ defined over Boolean variables $x_1,x_2,\ldots,x_n$
and an integer $M,$ we construct below a multi-user MIMO interference channel with $3n+m$ users, where the Boolean variable $x_i~(i=1,2,\ldots,n)$ corresponds to three users, including users $i, n+i,$ and $2n+i$; each clause $c_j~(j=1,2,\ldots,m)$ corresponds to user $3n+j$. %Therefore, there are $K=2n+m$ users and $N=2$ subcarriers in the constructed system.
Hence, ${\K}=\{1,2,\ldots,3n+m\}.$ %and $\N=\left\{1,2\right\}.$

{We now construct the direct-link and crosstalk channel matrices for all $3n+m$ users. All the direct-link channel matrices are set to be $$\bH_{k,k}=\bH_A,~k\in\K.$$
The corresponding crosstalk channel matrices are: for user $k,~k=1,2,...,n$,
set $$\bH_{k,n+k}=\bH_{k,2n+k}=\bH_B~\text{and}~\bH_{k,j}=\bm{0},~\forall~k\in{\cal K}\setminus\left\{k, n+k,2n+k\right\};$$
for user $k,~k=n+1,n+2,\ldots,2n,$ set
$$\bH_{k,n+k}=\bH_C~\text{and}~\bH_{k,j}=\bm{0},~\forall~k\in{\cal K}\setminus\left\{k, n+k\right\};$$
for user $k,~k=2n+1,2n+2,\ldots,3n,$ set
$$\bH_{k,k-n}=\bH_C~\text{and}~\bH_{k,j}=\bm{0},~\forall~k\in{\cal K}\setminus\left\{k-n, k\right\};$$
and for user $k,~k=3n+1, 3n+2,...,3n+m$, $\bH_{k,j}=\mathbf{0}$ for all ${\cal K}\setminus \left\{k\right\}$ except
\begin{equation}\label{matrix}\begin{array}{l} \bH_{k,j}=\left\{
\begin{array}{ll}\bH_A,&\mbox{if $\alpha_{\pi(k)}=x_j$ for some $j$;}\\
\bH_{B},&\mbox{if $\alpha_{\pi(k)}={\bar x}_j$ for some
$j,$}\end{array}\right.\\[15pt]
{\bH}_{k,j}=\left\{\begin{array}{ll}\bH_{A},&\mbox{if $\beta_{\rho(k)}=x_j$ for some $j;$}\\
\bH_{B},&\mbox{if $\beta_{\rho(k)}={\bar x}_j$ for some
$j,$}\end{array}\right.
\end{array}\end{equation}
   where $c_{k-3n}=\alpha_{\pi(k)}\vee \beta_{\rho(k)}$,
$\alpha$ and $\beta$ are taken from $\{x,{\bar x}\},$ and
$\pi$ and $\rho$ are mappings from $\{3n+1,3n+2,...,3n+m\}$ to
$\{1,2,...,n\}$.
} Set $\eta_k=1$ for all $k\in\K.$ Then, the SINR expressions of all users are: for $i=1,2,\ldots,n,$ %let
\begin{equation}\label{variablerate1-uv}
 \SI_{i}=\frac{\left|\bu^{\dagger}_i\bH_A\bv_i\right|^2}{1 + \left|\bu^{\dagger}_i\bH_B\bv_{n+i}\right|^2+ \left|\bu^{\dagger}_i\bH_B\bv_{2n+i}\right|^2},
 \end{equation}
 \begin{equation}\label{variablerate2-uv}
   \SI_{n+i}=\frac{\left|\bu^{\dagger}_{n+i}\bH_A\bv_{n+i}\right|^2}{1 + \left|\bu^{\dagger}_{n+i}\bH_C\bv_{2n+i}\right|^2},
 \end{equation}
 and
 \begin{equation}\label{variablerate3-uv}
 \SI_{2n+i}=\frac{\left|\bu^{\dagger}_{2n+i}\bH_A\bv_{2n+i}\right|^2}{ 1 + \left|\bu^{\dagger}_{2n+i}\bH_C\bv_{n+i}\right|^2};
\end{equation}
for $j=1,2,\ldots,m,$ %let
{\begin{equation}\label{clauserate-uv}\SI_{3n+j}=\left\{\!\!\!\!\!\!\!
  \begin{array}{ll}
&\displaystyle \frac{\left|\bu^{\dagger}_{3n+j}\bH_A\bv_{3n+j}\right|^2}{1 + \left|\bu^{\dagger}_{3n+j}\bH_A\bv_{i_1}\right|^2+ \left|\bu^{\dagger}_{3n+j}\bH_A\bv_{i_2}\right|^2},\,\text{if}~c_j=x_{i_1}\vee x_{i_2};\\[20pt]
&\displaystyle \frac{\left|\bu^{\dagger}_{3n+j}\bH_A\bv_{3n+j}\right|^2}{1 + \left|\bu^{\dagger}_{3n+j}\bH_A\bv_{i_1}\right|^2+ \left|\bu^{\dagger}_{3n+j}\bH_B\bv_{i_2}\right|^2},\,\text{if}~c_j=x_{i_1}\vee \bar x_{i_2};\\[20pt]
&\displaystyle \frac{\left|\bu^{\dagger}_{3n+j}\bH_A\bv_{3n+j}\right|^2}{1 + \left|\bu^{\dagger}_{3n+j}\bH_B\bv_{i_1}\right|^2+ \left|\bu^{\dagger}_{3n+j}\bH_A\bv_{i_2}\right|^2},\,\text{if}~c_j=\bar x_{i_1}\vee x_{i_2};\\[20pt]
&\displaystyle \frac{\left|\bu^{\dagger}_{3n+j}\bH_A\bv_{3n+j}\right|^2}{1 + \left|\bu^{\dagger}_{3n+j}\bH_B\bv_{i_1}\right|^2+ \left|\bu^{\dagger}_{3n+j}\bH_B\bv_{i_2}\right|^2},\,\text{if}~c_j=\bar x_{i_1}\vee \bar x_{i_2}.
  \end{array}\right.
\end{equation}}

Moreover, let $\gamma_{k}=1$ for all $k\in\K.$ Then, the constructed instance of problem \eqref{min-power-uv} is
\begin{equation}\label{constructed-uv}
\begin{array}{cl}
\displaystyle\min_{\left\{\bu_k,\,\bv_k\right\}} &\displaystyle \sum_{k\in\K}\left\|\bv_k\right\|^2\\[10pt]%\leq 2n+M+4(m-M)(\sqrt{2}-1)
\mbox{s.t.}&\displaystyle \SI_{k}\geq 1,\,k\in\cal K,\\[1pt]
&\displaystyle \|\bu_k\|^2 = 1,\,k\in\K,
\end{array}
\end{equation}where $\SI_k$ are given in \eqref{variablerate1-uv}, \eqref{variablerate2-uv}, \eqref{variablerate3-uv}, and \eqref{clauserate-uv}. %In the above, each user $k$ is associated with two variables $\bu_k$ and $\bv_k$ for all $k\in\K.$

Next, we show that there exists a truth assignment such that at least $M$
clauses are satisfied unanimously for the given MAX-2UNANIMITY
instance if and only if the optimal value of problem \eqref{constructed-uv} is less than or equal to $3n+2m-M.$

If there exists a truth assignment such that $M$ clauses in the
MAX-2UNANIMITY problem are unanimous, we claim that the optimal value of problem \eqref{constructed-uv} is less than or equal to $3n+2m-M.$ Let $\left\{x_i\right\}$ be the truth assignment such that $M$ clauses are unanimous in the MAX-2UNANIMITY problem. We set
$$\bu_i=\bu_{n+i}=\bu_{2n+i}=\bv_i=\bv_{n+i}=\bv_{2n+i}=\left(
                               \begin{array}{c}
                                 x_i \\
                                 1-x_i \\
                               \end{array}
                             \right)
,~i=1,2,\ldots,n.$$ With this, we can simply check that
$\SI_{k}\geq 1$ for all~$k=1,2,\ldots,3n.$ Furthermore, we consider SINR requirements of user $3n+j$ with $j=1,2,\ldots,m.$

\begin{itemize}
  \item [-] If the clause $c_j$ is unanimous, then we have either
  $$\SI_{3n+j}=\frac{\left|\bu^{\dagger}_{3n+j}\bH_A\bv_{3n+j}\right|^2}{1 + 2\left|\bu^{\dagger}_{3n+j}\be_1\right|^2}$$
  or $$\SI_{3n+j}=\frac{\left|\bu^{\dagger}_{3n+j}\bH_A\bv_{3n+j}\right|^2}{1 + 2\left|\bu^{\dagger}_{3n+j}\be_2\right|^2}.$$ In either cases, we can use a total transmission power of $1$ to make
  $\SI_{3n+j}\geq 1$ satisfied (by setting $\bu_{3n+j}=\bv_{3n+j}=\be_2$ in the former case and $\bu_{3n+j}=\bv_{3n+j}=\be_1$ in the latter case).
  \item [-] If the clause $c_j$ is not unanimous, then we must have
  $$\SI_{3n+j}=\frac{\left|\bu^{\dagger}_{3n+j}\bH_A\bv_{3n+j}\right|^2}{1 + \left|\bu^{\dagger}_{3n+j}\be_1\right|^2+\left|\bu^{\dagger}_{3n+j}\be_2\right|^2}.$$ In this case, we can use a total transmission power of $2$ to make $\SI_{3n+j}\geq 1$ satisfied (by setting $\bu={\sqrt{2}}\be/2$ and $\bv_{3n+j}=\be$).
\end{itemize}
As a result, if there exists a truth assignment such that at least $M$ clauses are satisfied unanimously, then the optimal value of problem \eqref{constructed-uv} is less than or equal to $3n+2m-M.$

For the converse part, assuming that the optimal value of problem \eqref{constructed-uv} is less than or equal to $3n+2m-M,$
we claim that at least $M$ clauses can be made unanimous. %according to Lemma \ref{lemma_small_power}, the optimal power allocation strategy to satisfy $R_{2i-1}\geq \ln2$ and $R_{2i}\geq \ln2$ is %at least $2$ and the power allocation vector must satisfying
%$$\left(p_{2i-1}^1,p_{2i-1}^2,p_{2i}^1,p_{2i}^2\right)^T=\left(1,0,0,1\right)$$ or
%$$\left(p_{2i-1}^1,p_{2i-1}^2,p_{2i}^1,p_{2i}^2\right)^T=\left(0,1,1,0\right).$$
%Since ``variable users'' and ``auxiliary variable users'' cause interferences to ``clause users'',
It follows from Lemma \ref{lemma_small_power_uv} {in Appendix \ref{app:lemma-uv}} that, for $i=1,2,\ldots,n,$ the optimal solution of problem \eqref{constructed-uv} must be
$$\bv_i\simeq\be_1~\text{or}~\bv_i\simeq\be_2.$$
%$$\bu_{i}\simeq\bu_{n+i}\simeq\bu_{2n+i}\simeq\bv_{i}\simeq\bv_{n+i}\simeq\bv_{2n+i}\simeq\be_1,$$ or
%$$\bu_{i}\simeq\bu_{n+i}\simeq\bu_{2n+i}\simeq\bv_{i}\simeq\bv_{n+i}\simeq\bv_{2n+i}\simeq\be_2.$$
 This, together with \eqref{clauserate-uv}, implies that the interference terms at user $3n+j$ must be $\left|\bu_{3n+j}^{\dagger}\be_1\right|^2$ and/or $\left|\bu_{3n+j}^{\dagger}\be_2\right|^2$ for all $j=1,2,\ldots,m.$ More specifically, there might be two cases:
\begin{itemize}
  \item [-] Case 1: the two interference terms at user $3n+j$ are the same, i.e., either both of them are $\left|\bu_{3n+j}^{\dagger}\be_1\right|^2$ or both of them are $\left|\bu_{3n+j}^{\dagger}\be_2\right|^2$;
  \item [-] Case 2: the two interference terms at user $3n+j$ are different, i.e., one is $\left|\bu_{3n+j}^{\dagger}\be_1\right|^2$ and the other is $\left|\bu_{3n+j}^{\dagger}\be_2\right|^2.$
\end{itemize}
If Case 1 happens for user $3n+j$, then the required total transmission power satisfying $\SI_{3n+j}\geq 1$ is at least $1;$ while if Case 2 happens for user $3n+j$, then the required total transmission power satisfying $\SI_{3n+j}\geq 1$ is at least $2.$ By the assumption that the optimal value of problem \eqref{constructed-uv} is less than or equal to $3n+2m-M,$ we know that Case 1 must happen at least $M$ times (Case 2 cannot happen more than $m-M$ times). Moreover, it can be checked that $$x_i=\left|v_{i1}\right|\in\left\{0,1\right\},~i=1,2,\ldots,n$$ is a truth assignment
which makes at least $M$ clauses in the MAX-2UNANIMITY problem satisfied unanimuously, where $v_{i1}$ is the first component of the vector $\bv_i$.

Finally, this transformation is in polynomial time. Since the MAX-2UNANIMITY
problem is NP-complete (cf. Lemma \ref{MAX-Full}), we conclude that the problem of checking the optimal value of problem \eqref{constructed-uv} is less than or equal to $3n+2m-M$ is strongly NP-hard. Hence, problem \eqref{min-power-uv} is strongly NP-hard. The proof of Theorem \ref{thm-min-power-uv} is completed.
%\end{IEEEproof}

}
%\cite{liu14tspofdma}
%\cite{IEEEexample:shellCTANpage}

%\bibliographystyle{./IEEEtran}
%\bibliography{./IEEEabrv,./IEEEexample}

%\bibliographystyle{IEEEtran}
\bibliographystyle{IEEEtran}
%\bibliography{IEEEabrv,D:/Study/paper2/bibtex/liuieeebibfiles}
\bibliography{IEEEabrv,liuieeebibfiles}

%\bibliography{IEEEfull,C:/Study/paper2/bibtex/liuieeebibfiles}

%\bibliographystyle{IEEEtran}
%\bibliography{IEEEabrv,C:/Study/paper2/bibtex/liubibfiles}

%IEEEabrv,
%
%\bibliography{C:/Study/paper2/bibtex/liubibfiles}

%\bibliography{liubibfiles}
%\bibliographystyle{spphys}
%\bibliographystyle{spmpsci}

%\begin{thebibliography}{1}
%
%\bibitem{coordinated} Y.-F. Liu, Y.-H. Dai, and Z.-Q. Luo,
%``Coordinated beamforming for MISO interference channel: Complexity
%analysis and efficient algorithms,'' \emph{IEEE Trans. Signal
%Process.,} vol. 59, no. 3, pp. 1142--1157, Mar. 2011.
%
%\bibitem{complexity} Z.-Q. Luo and S. Zhang, ``Dynamic spectrum management: Complexity and duality,'' \emph{IEEE J. Sel. Topics Signal Process.}, vol. 2, no. 1, pp.~57--73, Feb. 2008.
%\bibitem{boyd} S. Boyd and L. Vandenberghe, \emph{Convex Optimization,} New York, U.S.A.: Cambridge University Press, 2004.
%
%\end{thebibliography}
\end{document}